\documentclass[11pt]{article}
               \usepackage{listings}
               \usepackage{color}
               \usepackage{tocloft}
              \usepackage{amsthm}
\usepackage{amsmath}
\usepackage{array}
\usepackage{setspace}
 
\usepackage[T1]{fontenc}
\usepackage{natbib}
\setcitestyle{authoryear,open={(},close={)}}
\usepackage{authblk} 
\usepackage{mathrsfs}
\usepackage{apptools}
\AtAppendix{\counterwithin{lemma}{section}}
\usepackage{caption}
\usepackage{enumitem}
\usepackage{accents}

\usepackage[dvipsnames]{xcolor}
\usepackage{tcolorbox}

\usepackage[small]{titlesec}
\numberwithin{equation}{section} 
\titleformat{\subsubsection}
{\normalfont\bfseries}{\thesubsubsection}{1em}{}
\usepackage{adjustbox}
\usepackage{prodint}
\usepackage{booktabs}
\usepackage{bm}
\usepackage{bbold}
\usepackage{soul}
\usepackage{wasysym}
\usepackage[bottom=1.2in,top=1.2in,left=0.55in, right=0.55in]{geometry}
\newcommand{\EE}{\mathbb{E}}

\newcommand{\Obar}{\bar{O}}

\newcommand{\eps}{\varepsilon}

\newcommand{\F}{\mathcal{F}}

\newcommand{\Fj}{\mathcal{F}}

\newcommand{\R}{\mathbb{R}}
\newcommand{\N}{\mathbb{N}}
\newcommand{\1}{\mathbb{1}}

\newcommand{\cupdot}{\mathbin{\mathaccent\cdot\cup}}

\newtheorem{thm}{Theorem} 
\newtheorem{lemma}{Lemma} 

\usepackage{tikz}

\usetikzlibrary{arrows}
\tikzset{every picture/.style=remember picture}

\usepackage[utf8]{inputenc}
\usepackage[T1]{fontenc}
\usepackage{graphicx}
\usepackage{grffile}
\usepackage{longtable}
\usepackage{wrapfig}
\usepackage{rotating}
\usepackage[normalem]{ulem}
\usepackage{amsmath}
\usepackage{textcomp}
\usepackage{amssymb}
\usepackage{capt-of}
\usepackage{hyperref}
\title{Causal inference for calibrated scaling interventions on time-to-event processes} 
\author{Helene C. Rytgaard \& Mark J. van der Laan}
\date{\textit{Technical report} \\[0.5cm] \small Last updated: \today}

\begin{document}

\maketitle

\abstract{ This work develops a flexible inferential framework for
  nonparametric causal inference in time-to-event settings, based on
  stochastic interventions defined through multiplicative scaling of
  the intensity governing an intermediate event process.  These
  interventions induce a family of estimands indexed by a scalar
  parameter $\alpha$, representing effects of modifying event rates
  while preserving the temporal and covariate-dependent structure of
  the observed data generating mechanism.  To enhance
  interpretability, we introduce calibrated interventions, where
  \(\alpha\) is chosen to achieve a pre-specified goal, such as a
  desired level of cumulative risk of the intermediate event, and
  define corresponding composite target parameters capturing the
  downstream effects on the outcome process.
  This yields clinically meaningful contrasts while avoiding
  unrealistic deterministic intervention regimes.
  Under a nonparametric model, we derive efficient influence curves
  for $\alpha$-indexed, calibrated, and composite target parameters
  and establish their double robustness properties.  We further sketch
  a targeted maximum likelihood estimation (TMLE) strategy that
  accommodates flexible, machine learning based nuisance estimation.
  The proposed framework applies broadly to (causal) questions
  involving time-to-event treatments or mediators and is illustrated
  through different examples event-history settings.
  A simulation study demonstrates finite-sample inferential
  properties, and highlights the implications of practical positivity
  violations when interventions extend beyond observed data support.
}

\vspace{0.2cm}



\newpage

\begingroup
\setlength\parindent{0pt}
\setlength{\parskip}{0.4em}

\section{Introduction}
\label{sec:introduction}

In longitudinal and event history studies
\citep{andersen2012statistical}, a central aim is to understand how
time-to-event type processes, such as treatment initiation, disease
onset, or clinical interventions, affect outcomes like survival, death
due to specific causes, or disease progression.
Causal inference frameworks typically define interventions that are
static, assigning treatment uniformly (e.g., all subjects receive
treatment or control), dynamic, allowing treatment decisions to depend
on subject-specific characteristics, or stochastic, which more
generally assign treatment according to a user-specified probability
distribution \citep{diaz2013assessing,young2014identification,
haneuse2013estimation,van2007causal,hernan2006comparison,
chakraborty2013statistical,murphy2001marginal}.
In continuous-time settings, such interventions can be formalized
through modifications to the intensity functions governing treatment
or exposure processes \citep{roysland2011martingale,ryalen2020causal,
  rytgaard2021continuous,roysland2025graphical}. Static interventions
that prevent treatment initiation or discontinuation arise as special
cases, corresponding to `never treat' or `always treat' regimes.
However, medical decisions are guided by evolving clinical context and
are rarely applied homogeneously in practice.  As a result, static
interventions may be ill-defined or non-identifiable when treatment
eligibility varies over time, leading to severe violations of
positivity, and the reported estimands do not represent the effects of
clinical interest.
More broadly, researchers are often interested in describing how
time-to-event type treatments or other intermediate processes relate
to downstream outcomes, yet lack a practical and flexible intervention
framework for doing so in continuous-time settings.
In such settings, the goal is not always to prescribe individual-level
treatment rules, but to describe how systematic changes to
event-generating mechanisms would propagate to downstream outcomes.

In this work, we develop an inferential framework for nonparametric
causal inference in time-to-event settings, based on a natural class
of stochastic interventions which operate by scaling the intensity of
a counting process $N^z$, representing treatment initiation, disease
onset, surgery, or similar events, by a univariate scaling parameter
\(\alpha > 0 \). Rather than imposing deterministic or static rules,
these interventions proportionally change the instantaneous likelihood
of events, so post-intervention paths remain stochastic and respect
observed heterogeneity in clinical behavior.
This provides a flexible basis for analyzing realistic modifications
to time-to-event processes, with the intervention family interpolating
smoothly between observed practice and extreme regimes, and including
complete prevention as a limiting special case.
To enhance interpretability, we further introduce calibrated
interventions, in which $\alpha$ is chosen to achieve a pre-specified
goal, such as a desired level of cumulative risk of the intermediate
event, either fixed or relative to levels achieved under no
intervention. These calibrated interventions link the scaling
parameter \(\alpha\) to clinically meaningful targets and induce
composite causal parameters that capture downstream effects on the
outcome process, thereby linking flexible stochastic interventions to
interpretable causal questions.
We analyze the corresponding nonparametric estimation problems
targeting \(\alpha\)-indexed parameters, calibrated parameters, and
corresponding composite parameters.
We present the efficient influence curves for all target parameters
under a nonparametric statistical model, and derive and discuss their
double robustness properties. To make the theory operational, we
further sketch a targeted maximum likelihood estimator (TMLE) for the
fixed and calibrated targets that accommodate data-adaptive
machine-learning nuisance estimation.
The convenience of \(\alpha\)-scaling interventions becomes apparent
in several ways. First, estimation of weights for both inverse
probability weighting and targeted estimation greatly simplify:
because the scaled intensity is a multiplicative constant times the
original intensity, the likelihood ratio that defines weights reduces
to simple functions of cumulative event counts and cumulative
intensities, rather than ratios of intensities.  Furthermore, scaling
by any \(\alpha >0\) preserves the support of the original process, so
positivity holds for any \(\alpha>0\).  While extreme choices of
\(\alpha\) may in practice still induce large weights and numerical
instability, these practical issues are naturally accommodated within
the calibration framework, which provides a principled way to restrict
attention to plausible intervention regimes.

The intensity scaling interventions are connected to work on
incremental propensity score interventions
\citep{kennedy2019nonparametric}, as well as the causal mediation
framework based on parameter-indexed stochastic exposure interventions
proposed by \citet{diaz2020causal}.
However, our approach differs in several key ways. While
\cite{kennedy2019nonparametric} considers modification on the odds
ratio scale, we target the intensity scale directly, providing a more
natural specification in event history settings. Moreover, we
formulate and estimate the intervention effects in a more general data
setting, with right-censoring and competing risks and outcome events
allowed to happen in continuous or near continuous time.
We also emphasize that our proposed calibration framework is new,
linking intervention parameters to user-specified risk targets and
yielding composite parameters with a clear interpretation; this
calibration perspective could in turn be used to extend the methods of
\cite{kennedy2019nonparametric}.
%
%
%
Finally, we note that fixed-$\alpha$ intensity scaling coincides
mathematically with time-change representations that have appeared
previously in related continuous-time work \citep[][technical
report]{fawad2022hypothetical}, where counterfactual time changes were
used specifically to interpret treatment accelerations. Here we treat
stochastic intensity scaling itself as the primary intervention object
motivated by broader inferential goals, with interpretability provided
through calibration.
The nonparametric estimation framework developed here can nevertheless
be coupled with a time-change interpretation, thereby offering an
inferential extension of that perspective as well.

In summary, the contributions of this work are the following:
\begin{itemize}
\item[(i)] We formalize a class of stochastic intensity-scaling
  interventions for continuous-time event-history data and define
  corresponding \(\alpha\)-indexed target parameters.
\item[(ii)] We introduce calibrated interventions that select the
  scaling parameter \(\alpha\) to satisfy user-specified constraints
  on cumulative event risk, yielding interpretable composite target
  parameters, and offering a pragmatic analogue to indirect/direct
  effect decompositions.
\item[(iii)] Under a nonparametric model, we derive the efficient
  influence curves and establish efficiency theory for
  \(\alpha\)-indexed, calibrated, and composite targets.
\item[(iv)] We develop a targeted maximum likelihood estimation (TMLE)
  strategy that accommodates flexible, machine learning based nuisance
  estimation and avoids extreme weighting.
\end{itemize}
We further illustrate and motivate the proposed estimands through
several event-history examples drawn from real questions arising in
randomized trials and observational studies.

This document is structured as follows. Section \ref{sec:setting}
introduces the general setting and notation. Section
\ref{sec:interventions} presents our intervention framework: the
considered class of stochastic intensity interventions indexed by a
scalar parameter \(\alpha > 0\) and corresponding
intervention-specific (\(\alpha\)-indexed) target parameters,
calibrated interventions, calibration parameters and composite target
parameters.  Section \ref{sec:eic} studies the nonparametric
estimation problem and presents efficient influences curves. Section
\ref{sec:TMLE:estimation} sketches an estimation procedure based on
targeted maximum likelihood estimation. Section \ref{sec:sim:study}
presents a simulation study.  Section \ref{sec:discussion} concludes
with a discussion.

\section{General setting and notation} 
\label{sec:setting}

We consider an event history setting \citep{andersen2012statistical}
as follows. Suppose \(n\in\N\) subjects of a population are followed
over an interval of time \([0,\tau]\), each with observed data
characterized by a multivariate counting process \(N =
(N^{\ell},N^z,N^1,\ldots, N^J, N^c)\), with \(J\ge 1\), generating random times
\({T}_1 < T_2 < \cdots \) at which the disease or treatment status
(\(N^{\ell} ,N^z\)), an outcome of interest (\(N^1\)), competing risk
event status (\(N^2,\ldots, N^J\), if \(J\ge 2\)), and the censoring
status (\(N^c\)), may change. We shall focus on the case that the
outcome of interest \(N^1\) is an indicator of a particular disease to
happen, or an indicator of death due to a particular cause; extensions
to recurrent outcomes are possible but notationally heavier. We also
note that \(J=1\) if \(N^1\) is an indicator of all-cause mortality.
  Generally, \(N^z\) could represent a change in disease or treatment
  status, while \(N^{\ell}\) might capture, for example, disease
  status or a process indicating whether a biomarker has crossed a
  clinically relevant threshold. 
  Extensions to multiple such processes \(N^{\ell_1}, N^{\ell_2},...\)
  are straightforward, but simplified here for sake of presentation.
  We further consider baseline covariates \(L_0\in \R^d\) to be
  measured, and potentially a baseline treatment decision
  \(A_0\in\lbrace 0,1\rbrace\), to be made after observing
  \(L_0\). Note that \(A_0\) could for example indicate the
  randomization arm in a randomized controlled trial setting, or, if
  for example time zero is marked by a diagnosis in an observational
  study, \(A_0\) may represent the decision on a certain treatment
  option made following that diagnosis.
We collect the observed trajectory for one subject \(i\) in a bounded
interval \([0,t] \subseteq [0,\tau]\) as follows
\begin{align*}
  \bar{O}_i(t)  =\big( L_{0,i}, A_{0,i}, s,
  N_i^{\ell}(s), 
  N_i^z(s), N_i^1(s), \ldots, N_i^J(s) , N^c_i(s) \, :\, s \in\lbrace \tilde{T}_{i,r}\rbrace_{r=1}^{K_i(t)}\big), 
\end{align*}
with
\(K_i(t) = N^{\ell}_i(t) +N^{z}_i(t) + \sum_{j}^J N^{j}_i(t) +
N^c_i(t)\) denoting the number of events experienced at time \(t\).
We assume non-explosion on \([0,\tau]\), i.e., only finitely many
jumps on any compact interval.
We let \(T^d\) denote the survival time and \(T^c\) the censoring
time, so that we observe \( T^{\mathrm{end}} = \min (T^d, T^c)\) and
the indicator
\(\Delta = \1\lbrace T^d \le T^c\rbrace \sum_{j=1}^J j \cdot \1\lbrace
N^j(T^{\mathrm{end}}) = 1\rbrace \). We further let \(T^{\ell} ,T^z\)
denote the times of changes in disease/treatment status.
We thus have the collection of observed times
\((T_1, \ldots, T_{K })\) corresponding to the ordered version of the
event times \(( T^{\ell}, T^z, T^{\mathrm{end}})\), \(K:=K(\tau)\),
where it may be that \(T^{\ell}= \emptyset\) and/or
\(T^z = \emptyset\) for each particular subject. The observed counting
processes can also be represented as follows:
\( N^j(t) = \1\lbrace T^{\mathrm{end}}\le t, \Delta=j\rbrace
\in\lbrace 0,1 \rbrace, \,\, \text{for } j= 1,2,\ldots, J\),
\( N^c(t) = \1\lbrace T^{\mathrm{end}}\le t, \Delta=0\rbrace
\in\lbrace 0,1 \rbrace\),
\( N^{\ell}(t) = \1\lbrace T^{\ell} \le t, t \le T^{\mathrm{end}}
\rbrace \in \lbrace 0,1\rbrace\), and
\(N^{z}(t) = \1\lbrace T^{z} \le t, t \le T^{\mathrm{end}} \rbrace
\in \lbrace 0,1\rbrace\), with \( t\in [0,\tau]\), \(\tau >0\).
Note that we may equivalently write the observed data as
\begin{align*}
 \bar{O}_i(t) =(L_{0,i},A_{0,i},(T_{1,i},J_{1,i}),\ldots,(T_{K_i(t)},J_{K_i(t)})),
\end{align*}
where $J_k\in\mathcal E$ denotes the jump mark in the finite mark
space \(\mathcal{E}=\{1,\ldots,J,z,\ell,c\}\).
We let \(\F_{t} = \sigma( \bar{O}(t))\) denote the \(\sigma\)-algebra
generated by the observed data up until time \(t\), and 
let \(\mathcal{M}\) be the nonparametric model for the law of the
observed data \(O = \bar{O}(\tau)\) on the interval \([0,\tau]\).  For
any \(P\in\mathcal{M}\), we denote by
$\Lambda^x(t\mid\mathcal F_{t-})$ the predictable compensator for mark
\(x\), and, when absolute continuous, by
$\lambda^x(t\mid\mathcal F_{t-})$ its predictable intensity.
Each single predictable intensity may also be presented as a piecewise
(interval-specific) object, writing the intensity for mark $x$ in the
factorized form
\( {\lambda^x(t \mid \F_{t-})} = \sum_{k \ge 1} \1 \lbrace t \in
[T_{k-1}, T_k)\rbrace {\tilde{\lambda}^x_k (t \mid \Obar_{k-1} )}, \)
where, for each \(k\), $\tilde\lambda^x_k(\cdot\mid\Obar_{k-1})$ is a
predictable hazard function on $[T_{k-1},T_k)$ (the $k$-specific
hazard) and measurable with respect to the history $\Obar_{k-1}$. Note
that \(\tilde\lambda^x_k\) vanishes on intervals where mark \(x\) is
not admissible given \(\bar O_{k-1}\); for example
\(\tilde{\lambda}^z_k\) is only nonzero when no \(z\) event has
happened yet.

We further denote by \(\mu\) the density of the distribution of
baseline covariates \(L_0\in \R^d\) with respect to a dominating
measure \(\nu_L\), and by \(\pi\) the distribution of the baseline
treatment decision \(A_0 \in \lbrace 0,1\rbrace\). We denote by
\(P_0\in\mathcal M\) the true data-generating distribution and add the
subscript `0' to quantities under this law (e.g.,
\(\Lambda^x_0,\mu_0,\pi_0\)). The sample space for observed
paths that have exactly \(K\) jumps in \((0,\tau]\) is
\( \big(\mathbb
R^d\times\{0,1\}\times\{(t_1,j_1,\ldots,t_K,j_K):0<t_1<\cdots<t_K\le\tau,\
j_k\in\mathcal E\}\big)\). 
Each distribution \(P \in\mathcal{M}\) admits a general product
integral representation \citep{andersen2012statistical},
\begin{align}
  \begin{split}
   & dP (O) = \mu(L) d\nu_L(L) \pi(A \mid L)
    \Prodi_{s \le \tau}
\bigg( \prod_{x \in \mathcal{E}}
    \big({\Lambda}^x (dt \mid \F_{t-})\big)^{N^{x}(ds)}\bigg)  \bigg(1- \sum_{x \in \mathcal{E}} \Lambda^{x} (ds \mid \F_{s-}) \bigg)^{1-\sum_{x=1, \ldots, J,z,\ell,c} N^{x}(ds)}, 
  \end{split}\label{eq:factorize:P0}
\end{align}
with \(\prodi\) denoting the product integral \citep{gill1990survey}.
Note that the factor
\( \prodi_{s \le \tau} ( 1- \Lambda^\cdot (ds \mid
\F_{s-}))^{1-N^\cdot(ds)}\) evaluates to the exponential form
\( \exp (-\int_0^\tau \Lambda^\cdot (ds \mid \F_{s-}))\) when
\( \Lambda^\cdot\) is continuous.
Equivalently, when each compensator is absolutely continuous, the
density $p$ of the observed data under $P$ for a trajectory with
exactly $K$ jumps, with respect to the dominating measure
\(\nu_{L}\otimes\nu_A\otimes(\rho\otimes\nu_{\mathcal E})^K\), with
$\nu_A$ the counting measure on $\{0,1\}$, $\rho$ the Lebesgue measure
on $\mathbb{R}_+$, and $\nu_{\mathcal{E}}$ the counting measure on the
finite mark space $\mathcal{E} = \{1,\ldots,J,z,\ell,c\}$, can be
written on the form:
\begin{align*}
  \begin{split}
    p(o) &= \mu_0(\ell) \pi_0(a'\, \vert \, \ell)   \prod_{k=1}^K \bigg( \prod_{x\in \mathcal{E}}
           \big(    \tilde{\lambda}_{k}^x (t_k \mid \Obar_{k-1}) \big)^{\1 \lbrace j_k = x\rbrace}\bigg) 
 \exp \bigg( - \sum_{k=1}^{K+1}  \int_{t_{k-1}}^{t_k} \sum_{x\in \mathcal{E}}
           \tilde{\lambda}^x_{k}(t \mid \Obar_{k-1})dt \bigg)
           \bigg) 
           ,
  \end{split}
\end{align*}
with \(o = (\ell, a', t_{1}, j_{1}, \ldots, t_{K}, j_{K})\),
\(t_0:=0\) and \(t_{K+1} := \tau\).

\section{Intervention framework and target parameters}
\label{sec:interventions}

In this section we describe a class of stochastic interventions
tailored to the general event history setting described in Section
\ref{sec:setting}  where
intensity processes describe the rate at which certain events, such as
treatment initiation or disease onset, occur over time, conditional on
each subject’s observed history. These interventions specifically act
by multiplicative rescaling of the intensity \(\Lambda^z\) governing
the counting process \(N^z\) representing an intermediate
time-to-event process such as treatment initiation or disease onset.
The aim is twofold: to provide a flexible, interpretable family of
interventions indexed by a scalar \(\alpha>0\) and to define
calibrated/composite estimands that map policy-relevant targets (e.g.,
a desired cumulative incidence of the intermediate event) to an
\(\alpha\) level and then quantify the downstream impact on the
primary outcome. Below we first define the \(\alpha\)-scaling
interventions and the induced post-interventional law, and introduce
the primary and auxiliary target parameters. Then we present
calibration targets and composite parameters, discuss natural contrasts
and a direct/indirect-type decomposition, and state the main
interpretational and practical considerations. Further practical
discussion on how different causal contrasts can be formed and
interpreted within the proposed intervention framework can be found in
Supplementary Appendix B.

\subsection{\(\alpha\)-scaling stochastic interventions}
\label{sec:proposed:intervention}

We define the \(\alpha\)-scaling intervention which, for a rescaling
parameter \(\alpha > 0\), replaces intermediate process intensity
\(\Lambda^z\) as follows
\begin{align}
  \Lambda^z \mapsto \Lambda^{z,\alpha}= \alpha \Lambda^z,
  \label{eq:intensity:intervention}
\end{align}
while leaving the structural forms of other intensities unchanged.
The intervention alters the instantaneous rate at which
events occur without fixing event times or prescribing deterministic
decision rules; crucially, it preserves the dependence of
\(\Lambda^z\) on individual history, so heterogeneity in patient-level
risk remains intact under the intervention.
The intervention parameter \(\alpha > 0 \) controls the strength of
the hypothetical intervention, and can generally be interpreted
directly as a (conditional) hazard ratio. It includes the following
special cases, with \(\alpha = 1\) representing no intervention
(observed practice, or observed disease development), and
\(\alpha = 0\) representing complete prevention of the event
(equivalent to a ``censoring'' intervention).  Intermediate values
(e.g., \(\alpha = 0.5\)) represent proportional reductions in the
event rate, while values greater than one (e.g., \(\alpha = 1.5\))
correspond to increasing it.

\subsection{Post-interventional distribution}

We construct the post-interventional distribution $P^{a,\alpha}$
\citep[also known as the g-computation formula;][]{robins1986new} by
modifying the likelihood/path factorization \eqref{eq:factorize:P0} in
the following probabilistic manner:
\begin{enumerate}
\item
  Replace the treatment distribution \(\pi(a'\mid L_0)\) by the
  degenerate \(\delta_{a}(a')\) with all mass in
  \(a \in \lbrace 0,1\rbrace\) (i.e., enforce
  \(A_0=a \in \lbrace 0,1\rbrace\));
\item
  Remove right-censoring occurrence, by setting \(\Lambda^c = 0\); 
\item
  Replace \(\Lambda^z\) by \(\alpha\Lambda^z\) for \(\alpha > 0\).
\end{enumerate}
All other components (the law of baseline covariates $L_0$, the
functional forms linking histories to the other intensities) are kept
as under $P$.  We denote the resulting post-interventional
distribution by \(P^{a,\alpha}\).
In settings without a baseline treatment intervention,
interventions only involve no censoring and scaling \(\Lambda^z\) by
\(\alpha\Lambda^z\).
Note that \(\alpha = 1\) generally involves leaving \(\Lambda^z\) as
observed, so that, for example, \(P^{a,1}\) simply defines the
uncensored distribution under baseline treatment \(A_0=a\).

Since the intervention operates by multiplicatively scaling an
existing compensator, it preserves the support of the observed
data-generating process and does not introduce event types or
histories. As a result, the post-interventional distribution
\(P^{a,\alpha}\) is absolutely continuous with respect to \(P\),
ensuring that the positivity assumption holds by construction and that
the operation \(P \mapsto P^{a,\alpha}\) makes statistical sense and
gives rise to a well-defined measure over the same sample space.

\subsection{Intervention-specific target parameters}
\label{sec:target:parameter}

For fixed \(\alpha > 0\), a fixed time \(\tau >0\), and
\(a\in\lbrace 0,1\rbrace\), we define the intervention-specific target
parameter \(\Psi_1^{a,\alpha} \, : \, \mathcal{M} \rightarrow\R\) as
\begin{align}
  \Psi_1^{a,\alpha}(P)= \EE_{P^{{a,\alpha}}} [N^1(\tau)], 
  \label{eq:target:parameter}
\end{align}
and we similarly define
\( \Psi_1^{\alpha}(P)= \EE_{P^{{\alpha}}} [N^1(\tau)]\).
This mapping \(\alpha\mapsto \Psi_1^{a,\alpha}(P)\) defines a family
of intervention-specific parameters that quantify the expected outcome
under varying degrees of modification to the process \(N^z\). In
particular, contrasts such as
\( \Psi_1^{a,\alpha_1}(P) - \Psi_1^{a,\alpha_2}(P) \) represent the
difference in expected outcome under two alternative intervention
regimes, and may be interpreted as the effect of intensifying or
attenuating the rate of the event process \(N^z\) while fixing
treatment.

We further define the auxiliary parameter
\(\Psi_z^{a,\alpha} \, : \, \mathcal{M} \rightarrow\R\) capturing
the impact on the \(z\)-process itself
\begin{align}
  \Psi_z^{a,\alpha}(P)= \EE_{P^{{a,\alpha}}} [N^z(\tau)]. 
  \label{eq:target:parameter:z}
\end{align}
We similarly define
\(\Psi^{\alpha}_z(P)= \EE_{P^{{\alpha}}} [N^z(\tau)]\).  Thus, while
intervention with \(\alpha\) can be interpreted directly on the hazard
scale, the parameter in \eqref{eq:target:parameter:z} captures its
impact on the cumulative incidence of events of type \(z\), i.e., how
the intervention increases or decreases the probability of the event
of type \(z\) occurring over time. Notably, this parameter may be used
for calibrating \(\alpha\) to match a target level of incidence, or to
characterize trade-offs between intermediate and final outcomes; we
discuss this further in Section \ref{sec:impact:z}. First we state a
general result on the monotonicity and concavity of
\(\alpha \mapsto \Psi^{\alpha}_z(P)\).

\begin{lemma}
  The function \(\alpha \mapsto \Psi^{a,\alpha}_z(P)\) is increasing
  and concave in \(\alpha\), and strictly so when
  \(L^a(P):=P (\Lambda^z(\tau \mid \F^{a}_{\tau-})>0)>0\), where
  \(\F^{a}_{t}\) is the filtration generated by the observed data
  but evaluated in \(A_0=a\).  Moreover,
  \(\lim_{\alpha\rightarrow 0} \Psi^{a,\alpha}_z(P)=0\) and
  \(\lim_{\alpha\rightarrow\infty}\Psi^{a,\alpha}_z(P) = L^a(P) \);
  particularly
  \(\lim_{\alpha\rightarrow\infty}\Psi^{a,\alpha}_z(P) = 1 \) if and
  only if \(\Lambda^z(\tau \mid \F^{a}_{\tau-})>0\) almost surely.
  \label{lemma:shape:Psi:z:alpha}
\end{lemma}

The proof is found in Supplementary
Appendix
A.
We refer to $L^a(P)$ as defined in Lemma \ref{lemma:shape:Psi:z:alpha}
is the \textit{maximal population fraction} that can ever experience a type-$z$
event under any finite multiplicative scaling of the $z$-hazard.
In particular, subjects with
$\Lambda^z(\tau\mid\F^a_{\tau-})=0$ cannot be made to experience a
$z$-event by any finite $\alpha$.
The condition $L^a(P)<1$ indicates that for some covariate strata the
$z$-hazard is degenerate at zero; this is analogous to a violation of
positivity/support for an intervention that would create $z$-events in
those strata. Small $L^a(P)$ means the scaling intervention has
limited capacity to change the population $z$ risk.

\subsection{Calibrated interventions and composite target parameters}
\label{sec:impact:z}

While the intervention parameter \(\alpha > 0 \) can be interpreted
directly on the hazard scale, and the corresponding parameter
\(\Psi_z^{\alpha}(P)\) informs us about the cumulative incidence of
type \(z\) events under this intervention, it is often more meaningful
to target a specific value of \(\alpha\) that achieves a pre-specified
goal. For example, we may wish to determine how much the \(z\)
intensity should be scaled (via \(\alpha\)) to achieve a specific
level \(\theta\in (0,1)\) for the cumulative risk of the intermediate
event process.

To that end, we define different versions of calibration parameters
\(\alpha \, :\, \mathcal{M}\rightarrow\R\) defined as solutions to a
functional equations involving one or more \(\alpha\)-indexed
parameters. Below we give four useful choices of calibration
parameters:
\begin{align*}
  \text{Calibration towards a fixed level: } \,\,
  &\alpha^{a,\theta}(P) \;:=\; \big(\Psi_z^{a,\cdot}(P)\big)^{-1}(\theta),\qquad
    \theta \in (0,1);\\
  \text{Calibration towards an absolute change: } \,\,
  &\alpha^{a,\delta}(P) \;:=\; \big(\Psi_z^{a,\cdot}(P)\big)^{-1}(\Psi_z^{a,1}(P)+ \delta),
    \quad \delta \in (-1,1);\\
  \text{Calibration towards a relative change: } \,\,
  &\alpha^{a,\rho}(P) \;:=\; \big(\Psi_z^{a,\cdot}(P)\big)^{-1}(\rho\Psi_z^{a,1}(P)),
    \quad \rho >0.
  \\
  \text{Group-matching calibration: } \,\,
  &\alpha^{1-a}(P) \;:=\; \big(\Psi_z^{a,\alpha}(P)\big)^{-1}(\Psi_z^{1-a,1}(P)), \quad
    a \in \lbrace 0,1\rbrace.
\end{align*}
In words, these different rules answer related but distinct questions
about how to `tune' the \(z\)-intensity under treatment \(a\):
\(\alpha^{a,\theta}(P)\) finds the multiplicative scaling that
produces a prespecified absolute risk \(\theta\) of type \(z\) events;
\(\alpha^{a,\delta}(P)\) finds the scaling that produces an
\emph{absolute change} \(\delta\) from the baseline level at
\(\alpha=1\); \(\alpha^{a,\rho}(P)\) finds the scaling that produces a
\emph{relative} change (multiplicative factor) \(\rho\) from the
baseline at \(\alpha=1\); and \(\alpha^{1-a}(P)\) finds the \(\alpha\)
that makes the \(z\)-risk in treatment group/arm \(a\) equal the
(observed) \(z\)-risk in the other treatment group/arm.
Existence of solutions is discussed in Section \ref{sec:identification:remarks} below.

Once \(\alpha(P)\) is defined, we can evaluate other \(\alpha\)-indexed
parameters at this value to define a composite parameter, such as:
\begin{align*}
\Psi_1 (P) = \Psi_1^{\alpha(P)}(P). 
\end{align*}
This quantity represents the risk of the outcome of interest under an
intervention that, for example, scales the \(z\) event intensity just
enough to achieve the target level \(\theta \in [0,1]\) of type \(z\)
events. Importantly, the intervention indexed by \(\alpha(P)\)
modifies the intensity in a way that, while counterfactual, remains as
close as possible to the observed data-generating dynamics.

We further highligt that calibration enables a natural decomposition
of a baseline treatment contrast into a component attributable to
changes in the \(z\)-process, corresponding to a calibration-based
`indirect' part, and a residual `direct' part. Concretely, with
\(\alpha^{1-a}(P)\) as defined above we may write
\begin{align}
  \Psi_1^{a,1}(P) - \Psi_1^{1-a,1}(P)
  &= \underbrace{\big(\Psi_1^{a,1}(P) - \Psi_1^{a,\alpha^{1-a}(P)}(P)\big)}_{\substack{\text{calibration-based}\\\text{indirect component}}}
  + \underbrace{\big(\Psi_1^{a,\alpha^{1-a}(P)}(P) - \Psi_1^{1-a,1}(P)\big)}_{\substack{\text{calibration-based}\\\text{direct component}}}.
  \label{eq:indirect:direct:decomp:alpha*}
\end{align}
The first term quantifies how much of the arm-\(a\) outcome risk would
be altered if its rate of type \(z\) events were rescaled to match the
other arm's observed \(z\)-risk; the second term quantifies the
remaining difference once that matching is imposed.
We emphasize that this calibration-based decomposition of the effect
of the baseline treatment intervention differs from classical
mediation, where the mediator distribution (here with type \(z\)
events representing the mediator) is shifted across exposure
arms. Here, the `shifting' is instead governed by the scaling
parameter \(\alpha\), representing a more pragmatic form of
intervention suited to longitudinal or continuous-time settings, where
conventional interventions on the mediator distribution may not be
realistic.

\subsection{Causal reading and feasibility of calibration}
\label{sec:identification:remarks}

Statistically, the parameters \(\Psi_z^{a,\alpha}(P)\) and
\(\Psi_1^{a,\alpha}(P)\) are well defined for any \(P\in\mathcal
M\). A causal interpretation, i.e., reading the estimand as the
expected outcome under an intervention that only modifies the
\(z\)-mechanism and baseline treatment/censoring as specified,
requires additional conditions of no unmeasured confounding of the
baseline treatment \(A_0\) given \(L_0\) (when we intervene on
\(A_0\)), and further that the history recorded in \(\F_{t-}\)
captures all common causes of the \(z\)-process and outcome dynamics
(so the law-level modification \(\Lambda^z\mapsto\alpha\Lambda^z\)
corresponds to a well-defined intervention).

Moreover, all calibration choices define a value of \(\alpha\) by
inversion of the curve \(\alpha\mapsto\Psi_z^{a,\alpha}(P)\). By Lemma
\ref{lemma:shape:Psi:z:alpha}, this map is increasing and concave in
\(\alpha\), and thus its image
\(\{\,\Psi_z^{a,\alpha}(P):\alpha >0\,\}\) is an interval of the form
\((0,L^a(P))\).  Thus a prescribed target level is achievable by a
finite \(\alpha\) if that level falls within \((0,L^a(P))\).
For example, for group-matching calibration, a prescribed target
$\Psi_z^{1-a,1}(P)$ is achievable by a finite $\alpha$ only if
$\Psi_z^{1-a,1}(P)\in(0,L^a(P))$. If $\Psi_z^{1-a,1}(P)\ge L^a(P)$, no
finite scaling of the hazard in arm $a$ will attain the risk level
$\Psi_z^{1-a,1}(P)$. If the other arm's observed risk exceeds
$L^a(P)$, then matching by hazard scaling alone is impossible.
This would be a substantive finding in itself, as it indicates that
achieving such a target would require interventions that fundamentally
alter the underlying decision or eligibility mechanisms, and are
therefore incompatible with what is ever observed in practice,
corresponding to a direct violation of positivity.

\subsection{Illustration of application}
\label{sec:illustration:application}

We here give two concrete real-life inspired simulated examples to
build intuition about how the intervention-specific parameters
\(\Psi^{a,\alpha}_1(P)\) and the auxiliary parameters
\(\Psi^{a,\alpha}_z(P)\) behave as functions of the scaling factor
\(\alpha\), and to show how calibrated choices of \(\alpha(P)\) map to
composite outcome risks. An additional example and full details on the
data-generating distributions considered are provided in the
Supplementary Appendix C. 

\subsubsection*{Example 1: effect of an operation} 

Consider an observational setting in which patients diagnosed with a
serious disease may receive a surgical operation at a random time
during follow up (as represented by the process $N^z$). In many such
settings postponing operations increases mortality: hence reducing the
operation rate is potentially harmful.  A static intervention that
forces every patient either to always receive or to never receive
surgery is typically implausible, because surgical eligibility depends
on evolving disease severity, comorbidities and clinical judgement;
for some individuals surgery is effectively inevitable once clinical
criteria are met. The proposed intensity-scaling intervention
preserves these eligibility mechanisms while modifying the overall
rate at which operations occur. In particular, individuals with
$\Lambda^z(\tau\mid\F_{\tau-})=0$ remain unaffected by any finite
scaling, a fact reflected by the upper bound $L(P)$ introduced in
Lemma~\ref{lemma:shape:Psi:z:alpha}.
Figure \ref{fig:operation:simulation:illustration:function:alpha}
shows the true population curves \(\alpha \mapsto \Psi^{\alpha}_1(P)\)
and \(\alpha \mapsto \Psi^{\alpha}_z(P)\). The slope of the outcome
curve indicates whether increasing the \(z\)-rate tends to raise or
lower the outcome risk: here \(\alpha \mapsto \Psi_1^{\alpha}(P)\) is
decreasing, so reducing the rate of operations (smaller \(\alpha\))
increases the risk of dying.  Comparing the curves
\(\alpha \mapsto \Psi^{\alpha}_1(P)\) and
\(\alpha \mapsto \Psi^{\alpha}_z(P)\) makes the trade-off between
frequency (how many operations) and outcome effect (how mortality
changes) explicit, which is useful for policy decisions that must
balance a modest improvement in outcome against a potentially large
change in the number of operations performed.
As an example of a calibration target, define \(\alpha^{\rho}(P)\)
through \( \Psi_z^{\alpha^\rho}(P) = \rho\,\Psi_z^{1}(P), \) with
\(\rho=0.6\). Thus \(\alpha^{\rho}(P)\) is the scaling that reduces
the cumulative incidence of operations to \(60\%\) of its observed
level. The associated composite outcome is
\( \Psi_1^{\rho}(P) \;=\; \Psi_1^{\alpha^{\rho}(P)}(P).  \) In the
illustrative setting used for Figure
\ref{fig:operation:simulation:illustration:function:alpha} we obtain
\(\Psi_1^{\alpha=1}(P)=0.257\) and
\(\Psi_1^{\alpha^{\rho}}(P)=0.289\). In relative terms this
corresponds to a risk which is approximately \(12\%\) larger than the
observed risk, interpreted as a 40\% reduction in the number of
operations leading to 12\% more deaths.

\begin{figure}[!ht]
  \centering
  \includegraphics[width=0.7\textwidth]{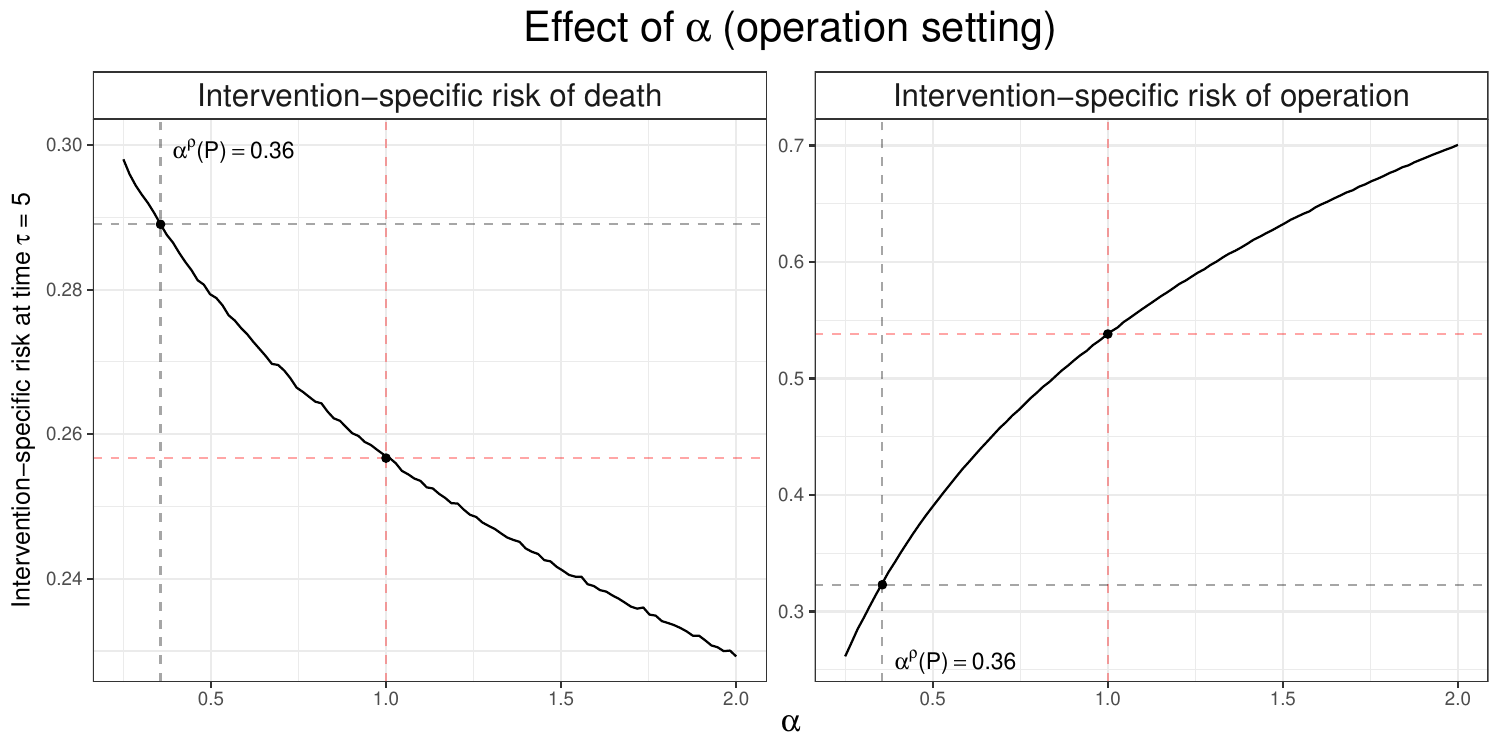}
  \caption{True values of the intervention-specific parameters for the
    operation example. The left plot shows
    \(\alpha \mapsto \Psi_1^{\alpha}(P)\); the right plot shows
    \(\alpha \mapsto \Psi_z^{\alpha}(P)\).  In the left plot, the
    difference between the curve and the horizontal red dashed line
    constitutes the contrast comparing the outcome risk when modifying
    the rate of operations by factor \(\alpha\) to the outcome risk
    with no modification of the rate of operation.  Comparing the
    curves on the left and the right plots reveals the trade-off
    between how frequently operations occur and their downstream
    impact on the outcome.  }
  \label{fig:operation:simulation:illustration:function:alpha}
\end{figure}

\subsubsection*{Example 2: Trial with rescue (drop-in) medication}
Consider a randomized trial in which subjects are allowed to initiate
a rescue medication (the $z$-event) during follow-up, but this
initiation mostly happens in the placebo arm. When initiation of
rescue medication can itself affect the outcome (e.g., mortality), a
key practical question is how much of the randomized treatment effect
operates through reductions in rescue initiation.
As in the previous example, forcing a homogeneous static rule
(everyone always/never initiates rescue) is typically unrealistic:
rescue initiation is driven by evolving biomarkers and symptoms (e.g.,
HbA1c in case of antidiabetic treatments) and eligibility depends on
patient-specific trajectories. The intensity-scaling intervention
preserves the clinical decision structure while proportionally
modifying the instantaneous hazard of rescue initiation. This allows
well-defined, interpretable intervention questions (e.g., reduce
rescue initiation by 30\%) without imposing clinically implausible
homogeneous rules.
To decompose the overall effect of randomized treatment, define
\(\alpha^{1}(P)\) as the scaling for the placebo arm \(a=0\) that
matches its \(z\)-risk to the observed treated-arm level. The
calibration-based decomposition is then
 \[
\Psi_1^{0,1}(P)-\Psi_1^{1,1}(P)
= \underbrace{\Psi_1^{0,1}(P) - \Psi_1^{0,\alpha^{1}(P)}(P)}_{\text{indirect via drop-in differences }}
+ \underbrace{\Psi_1^{0,\alpha^{1}(P)}(P) - \Psi_1^{1,1}(P)}_{\text{direct effect, balancing drop-in}}.
\]
The first term quantifies the portion of the treatment contrast
operating through changes in rescue initiation, and the second term is
the remaining direct
component. Figure~\ref{fig:dropin:simulation:illustration:function:alpha}
visualizes the curves \(\alpha \mapsto \Psi^{a,\alpha}_z(P)\) and
\(\alpha \mapsto \Psi^{a,\alpha}_1(P)\) for \(a=0,1\), as well as 
these segments. Of particular interest is here the direct effect
\(\Psi_1^{0,\alpha^{1}(P)}(P) - \Psi_1^{1,1}(P)\) reflecting the
effect of randomized treatment, had the overall initiation of rescue
mediation occurrence been matched between the treatment arms.

\begin{figure}[!ht]
\centering
 \begin{center}
    \includegraphics[width=0.7\textwidth,angle=0]{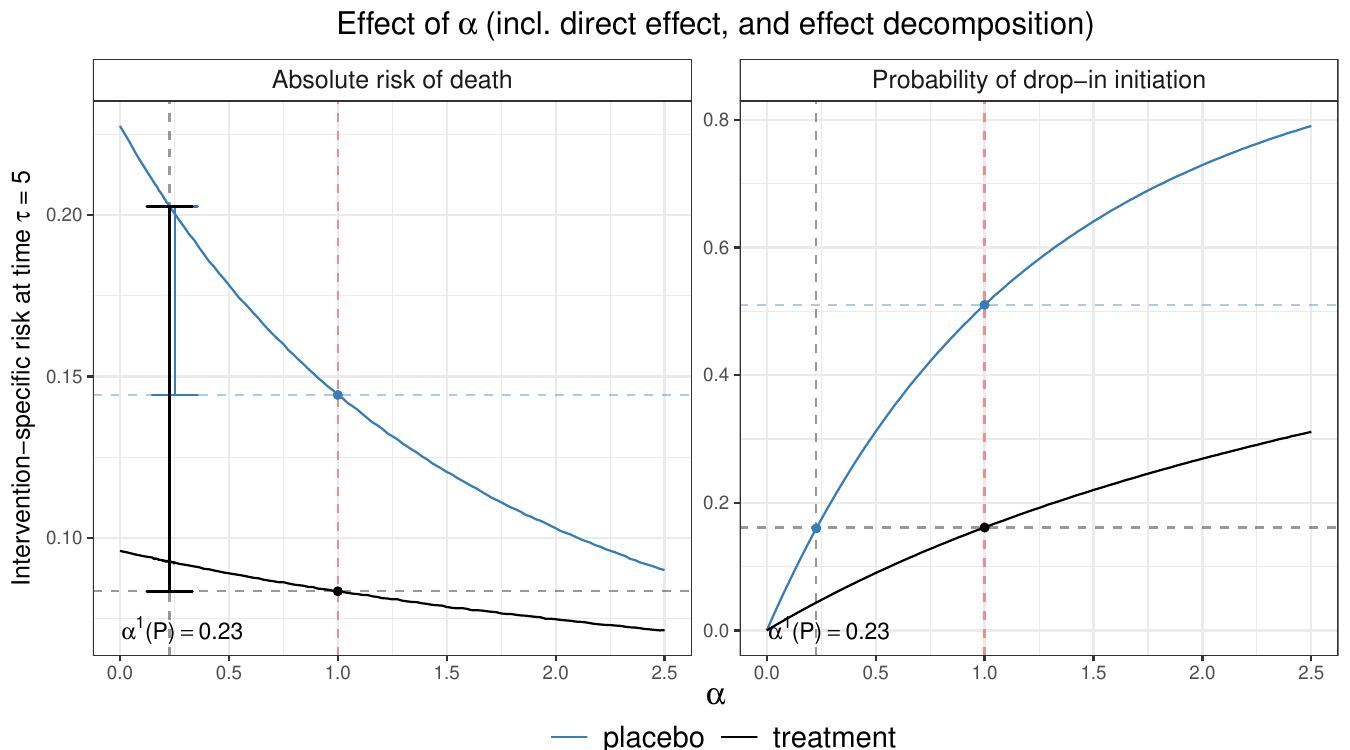}
  \end{center}
  \vspace{-0.5cm}
  \caption{True values of the intervention-specific parameters for the
    drop-in example. The left plot shows the intervention-specific
    absolute risks of death as function of \(\alpha\),
    \(\alpha \mapsto \Psi_1^{1,\alpha}(P)\) and
    \(\alpha \mapsto \Psi_1^{0,\alpha}(P)\); the right plot shows the
    intervention-specific absolute probability of drop-in initiation
    as function of \(\alpha\), \(\alpha \mapsto \Psi_z^{1,\alpha}(P)\)
    and \(\alpha \mapsto \Psi_z^{0,\alpha}(P)\). On the left plot the
    colored horizontal segments illustrate the indirect component
    $\Psi_1^{0,1}(P)-\Psi_1^{0,\alpha^{ 1}(P)}(P) = -0.0584$ (the
    length of the horizontal blue segment) and the direct component
    $\Psi_1^{0,\alpha^{ 1}(P)}(P)-\Psi_1^{1,1}(P) = 0.1192$ (the
    length of the horizontal black segment), which sum up to the total
    average treatment effect
    \(\Psi_1^{0,1}(P)-\Psi_1^{1,1}(P) = 0.0608\).}
   \label{fig:dropin:simulation:illustration:function:alpha}
\end{figure}

\section{Efficiency theory for estimation}
\label{sec:eic}

In this section we summarize the main theoretical ingredients that
underpin estimation and nonparametric inference for the target
parameters defined in Section~\ref{sec:interventions}.
We state the efficient influence curves for the \(\alpha\)-indexed
parameters
\(\Psi_x^{a,\alpha}(P) = \EE_{P^{{a,\alpha}}} [N^x(\tau)]\), for
\(x \in \lbrace 1,z\rbrace\), describe how these efficient influence
curves combine to yield efficient influence curves for the calibrated
and composite targets, and summarize the second-order remainders and
the resulting double robustness and rate requirements needed for
asymptotically linear estimation.  Full derivations and proofs are
given in the Supplementary Appendices D, F and E.

Since we will refer to several efficient influence curves, we briefly
explain our notation. We write
\( \phi^*_{\mathrm{functional}(), P}(O) \), where the first subscript
indicates the functional for which this is the efficient influence
curve, the argument \(P\in\mathcal{M}\) tells us the distribution at
which the curve is evaluated, and the superscript `\(*\)' signifies
that this is the efficient influence curve. Supplementary Material D
provides derivations of all presented efficient influence curves.

We begin by defining key components (clever covariates and weights)
used in the expression of the efficient influence curve for the
\(\alpha\)-indexed parameters. 
Clever covariates are defined as follows
(\(x' = 1, \ldots, J, \ell,z\))
\begin{align}
  \begin{split}
    h^{x,x'}_t (\lbrace\Lambda^j\rbrace_{j=1}^J,\Lambda^{\ell}, \Lambda^{z,\alpha}) (O)
    &=  \sum_{k=1}^K \1 \lbrace T_{k-1} < t \le T_k \rbrace \Big(
      \EE_{P^{{a,\alpha}}} \big[ N^x (\tau) \mid T_k = t, J_k = x',
      \F_{T_{k-1}}] \\[-0.3cm]
    &\qquad\qquad\qquad\qquad\qquad\qquad\qquad\qquad\qquad -\, \EE_{P^{{a,\alpha}}} \big[ N^x (\tau) \mid T_k >t , \F_{T_{k-1}}] \Big),
  \end{split}    \notag \\
        \begin{split}
          h^{x,z,\alpha}_t  (\lbrace\Lambda^j\rbrace_{j=1}^J,\Lambda^{\ell}, \Lambda^{z,\alpha}) (O)
          &  =  \alpha \, h^{x,z}_t (\lbrace\Lambda^j\rbrace_{j=1}^J,\Lambda^{\ell}, \Lambda^{z,\alpha}) (O),
        \end{split}   \notag
  \intertext{and clever weights \( w^{a,\alpha}_t (\pi, \Lambda^c, \Lambda^z) =  w^a_t (\pi, \Lambda^c)  w^{\alpha}_t (\Lambda^z)\) as}
  w^a_t (\pi, \Lambda^c) (O)& =   \frac{\delta_{a}(A) }{\pi (A\mid L) \prodi_{s < t} (1- \Lambda^c(ds \mid  \F_{s-}))}, \notag
  \\  
  w^{\alpha}_t (\Lambda^z) (O)& = 
                                \frac{ \Prodi_{s < t} \big( 
                                \Lambda^{z,\alpha} (ds \mid \F_{s-})\big)^{N^z(ds)}
                                \big(1- \Lambda^{z,\alpha} (ds \mid \F_{s-})\big)^{1-N^z(ds)}}{ \Prodi_{s < t} \big( \Lambda^z (ds \mid \F_{s-})\big)^{N^z(ds)}
                                \big(1- \Lambda^z (ds \mid \F_{s-})\big)^{1-N^z(ds)}} \notag \\
    & = 
      \alpha^{ N^z(t-)}   \frac{ \Prodi_{s < t}
      \big(1- \Lambda^{z,\alpha} (ds \mid \F_{s-})\big)^{1-N^z(ds)}}{ \Prodi_{s < t}
      \big(1- \Lambda^z (ds \mid \F_{s-})\big)^{1-N^z(ds)}},  \notag
\end{align}
where the factors involving the treatment distribution
\(\pi(a\mid L_0)\) or its intervened upon counterpart
\(\delta_{a}(a')\) vanish for the special case settings with no
baseline treatment.  Note that some of the clever covariates reduce
further to
\( h^{1,j}_t (\lbrace\Lambda^j\rbrace_{j=1}^J,\Lambda^{\ell},
\Lambda^{z,\alpha}) (O) = \sum_{k=1}^K \1 \lbrace T_{k-1} < t \le T_k
\rbrace \big( \1\lbrace j =1\rbrace - \, \EE_{P^{{a,\alpha}}} \big[
N^1 (\tau) \, \big\vert\, T_k >t, \Fj_{T_{k-1}}\big]\big)\), and
\( h^{z,j}_t (\lbrace\Lambda^j\rbrace_{j=1}^J,\Lambda^{\ell},
\Lambda^{z,\alpha}) (O) = \sum_{k=1}^K \1 \lbrace T_{k-1} < t \le T_k
\rbrace \big( N^{z}(T_{k-1}) - \, \EE_{P^{a,\alpha}} \big[ N^z (\tau)
\, \big\vert\, T_k>t, \Fj_{T_{k-1}}\big]\big)\), for
\( j=1,\ldots, J\).

\begin{thm}[Efficient influence curve for the \(\alpha\)-indexed
  parameters]
 The efficient influence curve for the intervention-specific parameter
 \(\Psi^{a,\alpha}_x \, : \, \mathcal{M}\rightarrow\R\), for given
 \(\alpha > 0\), can now be written as
\begin{align}
  & \phi^*_{\Psi_x^{a,\alpha}(),P} (O) \notag\\
  & =
    \sum_{j=1}^J\int_{t\le\tau}
    w_t^{a,\alpha} (\pi, \Lambda^c, \Lambda^z) (O)    h^{x,j}_t   (\lbrace\Lambda^j\rbrace_{j=1}^J,\Lambda^{\ell}, \Lambda^{z,\alpha}) (O)  \big( N^j(dt) - \Lambda^j(dt\mid  \F_{t-})\big)
            \label{eq:effi:if:1}
  \\
  & \,\, + \,   \int_{t\le\tau}
    w_t^{a,\alpha} (\pi, \Lambda^c, \Lambda^z) (O) h^{x,\ell}_t (\lbrace\Lambda^j\rbrace_{j=1}^J,\Lambda^{\ell}, \Lambda^{z,\alpha})(O)
    \big( N^{\ell}(dt) - \Lambda^{\ell}(dt \mid \F_{t-})\big)         \label{eq:effi:if:l} \\
  & \,\, + \,   \int_{t\le\tau}
    w_t^{a,\alpha} (\pi, \Lambda^c, \Lambda^z)
                        (O) h^{x,z, \alpha}_t (\lbrace\Lambda^j\rbrace_{j=1}^J,\Lambda^{z},
                        \Lambda^{z,\alpha})(O) \big( N^{z}(dt) - \Lambda^{z}(dt \mid
                        \F_{t-})\big)  \label{eq:effi:if:z} \\
  & \,\, + \,   \EE_{P^{{a,\alpha}}} [ N^x(\tau) \mid L_0] - \Psi_x^{a,\alpha}(P),       \label{eq:effi:if:mu} 
\end{align}
where the term in \eqref{eq:effi:if:z} is the specific contribution
that comes from \(\Lambda^z\) being unknown.
\label{thm:eic}
\end{thm}

We see immediately that the special case \(\alpha = 1\) corresponds to
\(\Lambda^{z,\alpha} = \Lambda^{z}\), i.e., not intervening on
\(\Lambda^z\), in which case
\( w^{a,\alpha}_t (\pi, \Lambda^c, \Lambda^z) = w^a_t (\pi,
\Lambda^c)\) and the contribution \eqref{eq:effi:if:z} corresponds to
the usual (when not intervened upon) intensity contribution to the
efficient influence curve. Furthermore, the special case
\(\alpha = 0\) corresponds to censoring events of type \(z\), in which
case the contribution \eqref{eq:effi:if:z} is zero and the numerator
of \( w^{\alpha}_t (\Lambda^z) (O)\) becomes \(1-N^z(t-)\).

We introduce notation for the derivatives of the \(\alpha\)-fixed
parameters:
\begin{align*}
(\Psi^{a}_{z,P})^{(1)}(\alpha) := \tfrac{d}{d\alpha}
\Psi^{a,\alpha}_z (P) , \quad \text{ and } \quad
(\Psi^{a}_{1,P})^{(1)} (\alpha) := \tfrac{d}{d\alpha}
  \Psi^{a,\alpha}_1 (P),
\end{align*}
where the superscript `\((1)\)' indicates the first derivative with
respect to \(\alpha\). By Lemma \ref{lemma:shape:Psi:z:alpha} we have
that \((\Psi^{a}_{z,P})^{(1)}(\alpha)>0\) as long as
\(L^a(P) := P(\Lambda^z(t\mid \F^{a}_{t-})>0)>0\) and the target
level for the curve \(\alpha \mapsto \Psi_z^{a,\alpha}(P)\) falls
inside \((0,L^a(P))\); this is important to ensure pathwise
differentiability of the \(\alpha(P)\) parameters, for which we
present efficient influences curves for in following lemma.

\begin{lemma}[General form of the efficient influence curves for
  calibration parameters.]
Let 
\[
\alpha^{a,\gamma}(P) = (\Psi_z^{a,\cdot}(P))^{-1}(\gamma(P))
\]
for a target functional \(\gamma \, :\, \mathcal{M}\to\mathbb{R}\). Then the
efficient influence curve of \(\alpha^{a,\gamma}\) is
\[
  \phi^*_{\alpha^{a,\gamma},P}(O) =
  \frac{1}{(\Psi^a_{z,P})^{(1)}(\alpha^{a,\gamma}(P))} \Big(
  \phi^*_{\gamma(),P}(O) -
  \phi^*_{\Psi_z^{a,\alpha^{a,\gamma}(P)}(),P}(O) \Big). 
\]
For the choices of calibration we have considered, we have: 
\begin{itemize}
\item[]  \textit{Fixed target:}
  \(\gamma(P) = \theta\), with
  \(0 < \theta < L^a(P)\), \(\phi^*_{\gamma(),P}=0\);
\item[] \textit{Absolute shift:}
  \(\gamma(P) = \delta + \Psi_z^{a,1}(P)\), with
  \(0 < \delta + \Psi_z^{a,1}(P)<L^a(P)\),
  \(\phi^*_{\gamma(),P}=\phi^*_{\Psi_z^{a,1}(),P}\);
\item[] \textit{Relative shift:} \(\gamma(P) = \rho \Psi_z^{a,1}(P)\),
  for \(0 < \rho \Psi_z^{a,1}(P) < L^a(P)\),
  \(\phi^*_{\gamma(),P}=\rho\phi^*_{\Psi_z^{a,1}(),P}\);
\item[] \textit{Cross-arm target:} \(\gamma(P) = \Psi_z^{1-a,1} (P)\),
  for \(\Psi^{1-a,1}_z (P) < L^a(P)\),
  \(\phi^*_{\gamma(),P}=\phi^*_{\Psi_z^{1-a,1}(),P}\).
\end{itemize}
\label{lemma:eic:alpha:theta}
\end{lemma}

\begin{thm}[Efficient influence curve of composite
  parameter.]
  The efficient influence curve for the composite parameter
  \(\Psi_1^a \, : \, \mathcal{M}\rightarrow\R\) defined as
  \(\Psi_1^a (P) = \Psi_1^{a,\alpha(P)} (P)\) with the parameter
  \(\alpha (P)\) defined as one of the choices of Lemma
  \ref{lemma:eic:alpha:theta} is
\begin{align}
  & \phi_{\Psi_1^a (),P}^{*}  (O) = \phi_{\Psi_1^{a,\alpha(P)}(),P}^*  (O) +  
    {(\Psi^{a}_{1,P})^{(1)} (\alpha(P))}\phi^*_{\alpha(),P} (O) ,
         \label{eq:effi:if:alpha:P}
\end{align}
where the relevant \(\alpha(P)\)-specific efficient influence curve is
substituted for \(\phi^*_{\alpha(),P}\).
\label{thm:eic:Psi:alpha:P}
\end{thm}

\subsection{Efficiency theorems for targeted substitution estimators}
\label{sec:targeted:inference:summary} 

Below we provide the main inferential results for targeted
substitution estimators for both \(\alpha\)-fixed parameters
$\Psi_x^{a,\alpha}(P)=\EE_{P^{a,\alpha}}[N^x(\tau)]$, for $x\in\{1,z\}$,
calibrated parameters $\alpha(P)$ obtained by inverting the map
$\alpha\mapsto\Psi_z^{a,\alpha}(P)$, and composite parameters
$\Psi_1^{a,\alpha(P)}(P)$ obtained by plugging $\alpha(P)$ into the
outcome functional.
These results are turned into operational estimators in Section
\ref{sec:TMLE:estimation}, where we describe a targeted maximum
likelihood estimation approach.
Below, Theorem~\ref{thm:eff:estimator:alpha:fixed} first establishes
asymptotic linearity of a targeted estimator for the parameter
$\Psi_x^{a,\alpha}(P)$ at a fixed $\alpha >0$ under standard nuisance
rate and empirical process conditions. Next, Lemma
\ref{lemma:eff:estimator:alpha:fixed} studies estimation for
calibrated parameters $\alpha(P)$.  Finally,
Theorem~\ref{thm:eff:estimator:composite} establishes asymptotic
linearity of the targeted estimator for the composite parameter
$\Psi_1^a(P) = \Psi_1^{a,\alpha(P)}(P)$. A brief comment on variance
estimation follows the theorems. Proofs are given in Supplementary
Appendix F.

\begin{thm}[Asymptotically linear estimation of the \(\alpha\)-fixed
  parameters.]
  Consider nuisance estimators
  \(\hat{P}_n =
  \{\hat{\pi}_n,\hat{\lambda}^c,\hat{\lambda}^z,\hat{\lambda}^{\ell},\hat{\lambda}^1,\ldots,\hat{\lambda}^J\}
  \in \mathcal{M}\). Assume that the following conditions hold:
  \begin{enumerate}
\item[1a)] \textit{Nuisance-rate conditions:} The nuisance components
  \(\eta\in\{\pi,\lambda^c,\lambda^z,
  \lambda^{\ell},\lambda^1,\ldots,\lambda^J\}\) are estimated by
  \(\hat\eta_n\) for which
  \((\sum_{\eta_1 \in \{\pi,\lambda^c,\lambda^z\}}\|\hat\eta_{1,n} -
  \eta_{1,0}\|_{L^2(P^{a,\alpha}_0)}) (\sum_{\eta_2 \in \{\lambda^z,
    \lambda^{\ell},\lambda^1,\ldots,\lambda^J\}}\|\hat\eta_{2,n} -
  \eta_{2,0}\|_{L^2(P^{a,\alpha}_0)})=o_P(n^{-1/2})\);\vspace{0.2cm}
\item[1b)] \textit{Empirical process control:} The class
  \(\{\phi^*_{\Psi_x^{a,\alpha}(),P}:P\in\mathcal M\}\) is
  $P_0$-Donsker, and
  \(\|\phi^*_{\Psi_x^{a,\alpha}(),\hat{P}_n}-\phi^*_{\Psi_x^{a,\alpha}(),P_0}\|_{L^2(P_0)}\overset{P}{\to}0\);
\end{enumerate}
and that the estimator \(\hat{P}_n\) solves the efficient influence
curve equation
\begin{align}
  \mathbb{P}_n \phi^*_{\Psi_x^{a,\alpha}(),\hat{P}_n} =
  o_P(n^{-1/2}). 
  \label{eq:key:EIC:eq}
\end{align}
It follows that the remainder is asymptotically negligible,
\begin{align}
  R_{\Psi^{a,\alpha}_x()} (\hat{P}_n, P_0) & = \Psi_x^{a,\alpha} (\hat{P}_n) - \Psi_x^{a,\alpha} (P_0) +
                                               P_0 \phi^*_{\Psi_x^{a,\alpha}(),\hat{P}_n} = o_P( n^{-1/2}),
                            \label{eq:remainder:alpha:x:negligible}
\end{align}
and,
\begin{align}
  \Psi_x^{a,\alpha} (\hat{P}_n) -    \Psi_x^{a,\alpha} (P
  _0) =\mathbb{P}_n \phi^*_{\Psi_x^{a,\alpha}(),P_0} + o_P(n^{-1/2});
  \label{eq:asympt:linearity} 
\end{align}
that is, \(\hat{\psi}_n =  \Psi_x^{a,\alpha} (\hat{P}_n)\) is asymptotically
linear at \(P_0\) with influence function equal to the efficient
influence curve \(\phi^*_{\Psi_x^{a,\alpha}(),P_0}\).
\label{thm:eff:estimator:alpha:fixed}
\end{thm}

The form of the rate condition in Assumption 1a encodes the double
robustness properties of the estimation problem, and indicates that
consistency is achieved whenever either the nuisance group
\(\{\pi,\lambda^c,\lambda^z\}\) or the nuisance group
\(\{\lambda^z, \lambda^{\ell},\lambda^1,\ldots,\lambda^J\}\) is
estimated consistently.
The sufficiency of the rate condition as stated follows directly from
Lemma E.1 (Supplementary Material E), which gives the second-order
remainder for the target parameter \(\Psi_x^{a,\alpha}(P)\).
Inspection of that remainder shows that it depends only on intensities
inside integrals, so the actual necessary convergence requirements are
likely weaker than those stated in Assumption 1a.
Because \(\lambda^z\) appears in both groups, its estimation error
enters both factors of the product-form remainder and therefore cannot
be misspecified. A sufficient condition for asymptotic linearity is
still that
\(\|\hat\lambda^z_n-\lambda^z\|_{L^2(P^{a,\alpha}_0)} =
o_P(n^{-1/4})\), the usual requirement in double robustness settings
when relying on machine learning nuisance estimators rather than
parametric rates. Moreover, under \(\alpha\)-scaling its weight
contribution is in fact driven by error in estimating the cumulative
intensity, not by ratios of intensities.

\begin{lemma}[Asymptotically linear estimation of calibrated
  \(\alpha(P)\).] Consider
  \(\alpha^{a,\gamma} \, : \, \mathcal{M}\rightarrow\R\) defined as in
  Lemma \ref{lemma:eic:alpha:theta}: 
  \[\alpha^{a,\gamma}(P) =
    (\Psi_z^{a,\cdot}(P))^{-1}\big(\gamma(P)\big),\] with
  \(\gamma:\mathbb{R}^2\to\mathbb{R}\) given as one of the four
  examples of Lemma \ref{lemma:eic:alpha:theta} and defined such that
  \(\gamma(P)\in [\gamma_1, \gamma_2]\) with
  \(0 < \gamma_1 \le \gamma_2 < L^a(P)\).

  When needed (depending on the choice of calibration target), define
  \(\hat{\gamma}_n = \gamma (\hat{P}_n)\), where \(\hat{P}_n\) is an
  estimator fulfilling Assumptions 1a and 1b from Theorem
  \ref{thm:eff:estimator:alpha:fixed} for \(x=z\) and \(\alpha=1\) and
  which solves the efficient influence curve equation
  \(\mathbb{P}_n \phi^*_{ \gamma(),\hat{P}_n} = o_P(n^{-1/2})\).
  For any of the choices of \(\gamma\) from Lemma
  \ref{lemma:eic:alpha:theta}, it follows from Theorem
  \ref{thm:eff:estimator:alpha:fixed}, that \(\hat{\gamma}_n\) is
  asymptotically linear with influence function equal to the efficient
  influence curve \(\phi^*_{\gamma(),P_0}\).

  Define
  \(\hat{\alpha}_n = \alpha(\hat{P}_n) = (\Psi^{a,\cdot}_z
  (\hat{P}_n))^{-1} (\hat{\gamma}_n)\),
  where \(\hat{P}_n\) is an estimator fulfilling Assumptions 1a and 1b
  from Theorem \ref{thm:eff:estimator:alpha:fixed} for \(x=z\) and
  \(\alpha = \hat{\alpha}_n \), and which solves \eqref{eq:key:EIC:eq}
  at \(\alpha = \hat{\alpha}_n \):
\begin{align}
  \mathbb{P}_n \phi^*_{\Psi_x^{a,\hat{\alpha}_n}(),\hat{P}_n} =
  o_P(n^{-1/2}). 
  \label{eq:key:EIC:eq:alphahat}
\end{align}  
Further assume the following:
  \begin{enumerate}
  \item[2a)] \(L^a(P)>0\) for each \(P\in\mathcal{M}\), so that
    \(\alpha\mapsto\Psi^{a,\alpha}_z(P)\) is strictly increasing and
    concave.
\item[2b)] The first and second derivatives of
  $\alpha\mapsto \Psi^a_{z,P}(\alpha)$ exist and are uniformly bounded
  in a neighborhood of $(\alpha_0,P_0)$, where
  \(\alpha_0 :=\alpha^{a,\gamma}(P_0)\).
\item[2c)] {The class
    \( \{\, f_z(\alpha,P):\alpha\in [\alpha_1,\alpha_2],\;
    P\in\mathcal{M}\,\}\),
    \(\alpha_1 = (\Psi_z^{a,\cdot}(P))^{-1}(\gamma_1)\),
    \(\alpha_2 = (\Psi_z^{a,\cdot}(P))^{-1}(\gamma_2)\), with
    \( f_z(\alpha,P)(O):=\phi^*_{\Psi_z^{a,\alpha}(),P}(O)\), is
    \(P_0\)-Donsker, and, whenever
    \((\hat{\alpha}_n,\hat{P}_n)\overset{P}{\to}(\alpha_0,P_0)\), we
    have
    \(\|f_z(\hat{\alpha}_n,\hat{P}_n)-f_z(\alpha_0,P_0)\|_{L^2(P_0)}\overset{P}{\to}0\).}
  \end{enumerate}
  Then it holds that
  \begin{align}
  \mathbb{P}_n    \phi^*_{\alpha^{a, \gamma}(),\hat{P}_n} =  o_P(n^{-1/2}), 
    \label{eq:key:EIC:eq:alpha:theta}
\end{align}
and
\begin{align}
  \alpha(\hat{P}_n) -   \alpha (P   _0) =\mathbb{P}_n \phi^*_{\alpha^{a, \gamma}(),P_0} + o_P(n^{-1/2}),
  \label{eq:asympt:linearity:alpha:theta} 
\end{align}
that is, \(\hat{\alpha}_n = \alpha(\hat{P}_n)\) is asymptotically
linear at \(P_0\) with influence function equal to the efficient
influence curve \(\phi^*_{\alpha^{a, \gamma}(),P_0}\) presented in
Lemma \ref{lemma:eic:alpha:theta}. Note that for the purpose of
\eqref{eq:asympt:linearity:alpha:theta}, it suffices that
\(\Psi_z^{a,\hat{\alpha}^*_n}(\hat{P}_n)= \hat{\gamma}_n +
o_P(n^{-1/2})\); exact equality is not needed.
\label{lemma:eff:estimator:alpha:fixed}
\end{lemma}

We remark that Assumptions 2c, and 3b which follows below, are not
expected to be much stronger in practice than Assumption 1b, provided
that $[\alpha_1,\alpha_2]$ restricts to a range avoiding instabilities
such as exploding clever weights.

\begin{thm}[Asymptotically linear estimation of composite parameter.]
  Construct \(\hat{\alpha}_n=\alpha(\hat{P}_n)\) as in Lemma
  \ref{lemma:eff:estimator:alpha:fixed}. Consider the composite
  parameter \(\Psi_1^a(P)=\Psi_1^{a,\alpha(P)}(P)\) and its estimator
  \(\hat\psi_{1,n}=\Psi_1^{a,\hat{\alpha}_n}(\hat P_{n})\), such that
\begin{align}
  \mathbb{P}_n \phi^*_{\Psi_1^{a,\hat{\alpha}_n}(),\hat{P}_{n}} =
  o_P(n^{-1/2}).
  \label{eq:key:composite:EIC:eq}
\end{align}
Under Assumptions 1a and 1b from Theorem
\ref{thm:eff:estimator:alpha:fixed}, Assumptions 2a, 2b and 2c from Lemma
\ref{lemma:eff:estimator:alpha:fixed}, as well as the following
conditions:
\begin{enumerate}
\item[3a)] The first and second derivatives of
  $\alpha\mapsto \Psi^a_{1,P}(\alpha)$ exist and are uniformly bounded
  in a neighborhood of $(\alpha_0,P_0)$.
\item[3b)] {The class
    \( \{\, f_1(\alpha,P):\alpha\in [\alpha_1,\alpha_2],\;
    P\in\mathcal{M}\,\}\), with
    \( f_1(\alpha,P)(O):=\phi^*_{\Psi_1^{a,\alpha}(),P}(O)\), is
    \(P_0\)-Donsker, and, whenever
    \((\hat{\alpha}_n,\hat{P}_n)\overset{P}{\to}(\alpha_0,P_0)\), we
    have
    \(\|f_1(\hat{\alpha}_n,\hat{P}_n)-f_1(\alpha_0,P_0)\|_{L^2(P_0)}\overset{P}{\to}0\).}
  \end{enumerate}
  it holds that
\[
\hat\psi_{1,n} - \Psi_1^{a}(P_0)
= \mathbb{P}_n \,\phi_{\Psi_1^a(),P_0}^* + o_P(n^{-1/2}),
\]
with the efficient influence curve as defined in Theorem
\ref{thm:eic:Psi:alpha:P}.
\label{thm:eff:estimator:composite}
\end{thm}

A straightforward consequence of each lemma/theorem is that we can use
the asymptotic normal distribution
\( \sqrt{n}\, \big( \hat{\psi}^*_n - \psi_0 \big)
\overset{\mathcal{D}}{\rightarrow} \mathcal{N} (0, {P}_0
\phi^*_{\Psi(),{P}_0} )^2) \) following from asymptotic linearity to
provide an approximate two-sided confidence interval. The asymptotic
variance of the estimator is given from the variance of the efficient
influence function and can be estimated by \(\hat{\sigma}_n^2 /n\)
where
\(\hat{\sigma}_n^2 = \mathbb{P}_n (\phi^*_{\Psi(),\hat{P}_n} )^2\).

\section{Targeted maximum likelihood estimation}
\label{sec:TMLE:estimation}

The theoretical analysis of the previous section shows the form of the
efficient influence curves and double-robustness structure for the
$\alpha$-indexed and calibrated targets.  In this section, we present
a targeted maximum likelihood estimation (TMLEs) procedure that
accepts both flexible machine learning based and (semi)parametric
regression initial fits for the predictable intensities and the
baseline treatment model, updates those fits to solve the efficient
influence equations, and inverts these to obtain estimates of
calibrated parameters, and targets final estimators for the composite
parameters. 
Formal regularity conditions and asymptotic results are as summarized
in Section~\ref{sec:targeted:inference:summary} and the technical
implementation details are deferred to Supplementary Appendix G. Here
we first give a high-level overview of the main components for
constructing and providing inference for targeted estimators:

\begin{enumerate}
\item  \textit{Construction of estimators for clever weights and clever
    covariates corresponding to a given \(\alpha\)-fixed target;
    plug-in evaluation of that same target.} Given initial nuisance
  estimators
  \(\hat P_n=(\hat{\Lambda}_n^{1},\ldots,\hat{\Lambda}_n^{J},
  \hat{\Lambda}_n^{\ell},\hat{\Lambda}_n^{z},\hat{\Lambda}_n^{c},\hat{\pi}_n)\)
  of the predictable intensities and of the baseline treatment
  distribution, we describe an algorithm for estimation of clever
  covariates and evaluation of the g-computation formula (overview in
  Section \ref{sec:clever:covar:estimation}), and estimation of clever
  weights (Section \ref{sec:estimation:clever:weights}).
\item \textit{Estimation of the \(\alpha\)-fixed parameters
    \(\Psi_z^{a,\alpha} (P)\) and \(\Psi_1^{a,\alpha}(P)\).} We
  describe a targeted maximum likelihood estimation procedure
  (combining Sections
  \ref{sec:initial:estimation}--\ref{sec:estimation:TMLE:step}) for
  estimation of an \(\alpha\)-fixed parameter
  \(\Psi_x^{a,\alpha} (P)\), \(x \in \lbrace 1,z\rbrace\). This
  procedure provides estimators which solve
  \begin{align*}
\mathbb{P}_n \phi^*_{\Psi_x^{a,\alpha}(),\hat{P}_n^{*}} =
    o_P(n^{-1/2}),
  \end{align*}
  with
  \(\hat{P}_n^{*}= (\hat{\Lambda}^{1}_{n,m^*}, \ldots,
  \hat{\Lambda}^{J}_{n,m^*},
  \hat{\Lambda}^{\ell}_{n,m^*},\hat{\Lambda}^z_{n,m^*},
  \hat{\Lambda}^c_n, \hat{\pi}_n)\), for fixed \(\alpha \ge 0 \).
  Theorem \ref{thm:eff:estimator:alpha:fixed} establishes inference
  for \(\hat{\psi}^{a,\alpha,*}_{x,n} = \Psi_x^{a,\alpha} (\hat{P}_n^*)\),
  and the variance can be estimated by the empirical variance of the
  estimated efficient influence curve,
  \(\hat{\sigma}_n^2 = \mathbb{P}_n (
  \phi^*_{\Psi_x^{a,\alpha}(),\hat{P}_n})^2 /n \).
\item \textit{Estimation of calibration parameters \(\alpha(P)\).}
  This involves the inverse of \(\alpha \mapsto \Psi_z^{a,\alpha}(P)\)
  and for inference also the derivative
  \((\Psi^{a}_{z,P})^{(1)}(\alpha) := \tfrac{d}{d\alpha}
  \Psi^{a,\alpha}_z (P)\). In practice, we replace
  \(\Psi_z^{a,\alpha}(P)\) by its targeted estimator
  \(\hat{\psi}^{a,\alpha,*}_{z,n}\), when relevant compute the
  targeted estimator \(\hat{\gamma}_n^* =\gamma(\hat{P}_n^*)\), and
  define \(\hat{\alpha}_n^*\) as the solution to
  \(\hat{\psi}^{a,\alpha,*}_{z,n} = \hat{\gamma}_n^*\) (a concrete
  algorithm for finding the inverse is provided in Supplementary
  Appendix G). Inference for
  \(\hat{\alpha}^{*}_{n}= \alpha(\hat{P}_n^*)\) follows from Lemma
  \ref{lemma:eff:estimator:alpha:fixed}, and
  the variance can be estimated by the empirical variance of the
  estimated efficient influence curve,
  \(\hat{\sigma}_n^2 = \mathbb{P}_n (\phi^*_{\alpha(),\hat{P}_n})^2 /n
  \). For variance estimation, a consistent estimator of the
  derivative \((\Psi^{a}_{z,P})^{(1)}(\alpha(P)) \) is needed, which
  can be achieved with a simple difference estimator (detailed in
  Supplementary Appendix G).
\item \textit{Estimation of the composite parameter
    \(\Psi_1^a(P) = \Psi_1^{a,\alpha(P)}(P)\).} We achieve this by
  plugging in the estimator \(\hat{\alpha}_n^* = \alpha(\hat{P}_n^*)\)
  from 2., and subsequently using the targeting procedure from
  1. Theorem \ref{thm:eff:estimator:composite} establishes inference
  for the estimator
  \(\hat{\psi}^{*}_{1,n} = \Psi^{a,\hat{\alpha}_n^*}_1 (\hat{P}_n^*)\)
  which follows under the same conditions as in 1. and 2., where
  variance estimation, using the empirical variance of the estimated
  efficient influence curve,
  \(\hat{\sigma}_n^2 = \mathbb{P}_n ( \phi_{\Psi_1^a
    (),\hat{P}_n}^{*})^2 /n \), also requires that the derivative
  \((\Psi^{a}_{1,P})^{(1)}(\alpha(P)) \) is estimated consistently
  which can be based on the same approach as for
  \((\Psi^{a}_{z,P})^{(1)}(\alpha(P)) \).
\end{enumerate}

\subsection{Initial estimation}
\label{sec:initial:estimation} 

The estimation algorithms that follow refine initial intensity
estimators to, in the end, achieve valid inference for the target
parameters. Such initial estimators of intensities may be obtained by
classical Cox regression, highly adaptive lasso, or other flexible
learners that can be written on the following form:
\begin{align*}
  \begin{split}
    \Lambda^j (dt \mid \F_{t-})
    & = 
      \1\lbrace  T^{\mathrm{end}} \ge t\rbrace
    \exp ( f^j(t, \kappa^{\ell}(\bar{N}^{\ell}(t-)),\kappa^{z}(\bar{N}^{z}(t-)), A_0,L_0)) dt, \text{ for } j=1,\ldots, J, \\
    \Lambda^{\ell} (dt \mid \F_{t-})
    & = 
      \1\lbrace  T^{\mathrm{end}} \ge t, T^{\ell} \ge t\rbrace
      \exp ( f^{\ell}(t, \kappa^{\ell}(\bar{N}^{\ell}(t-)),\kappa^{z}(\bar{N}^{z}(t-)), A_0,L_0)) dt, \\
    \Lambda^{z} (dt \mid \F_{t-})
    & = 
      \1\lbrace  T^{\mathrm{end}} \ge t, T^{z} \ge t\rbrace
      \exp ( f^{z}(t, \kappa^{\ell}(\bar{N}^{\ell}(t-)),\kappa^{z}(\bar{N}^{z}(t-)), A_0,L_0)) dt , \\
    \Lambda^c (dt \mid \F_{t-})
    & = 
      \1\lbrace  T^{\mathrm{end}} \ge t\rbrace
     \exp ( f^c(t, \kappa^{\ell}(\bar{N}^{\ell}(t-)),\kappa^{z}(\bar{N}^{z}(t-)), A_0,L_0)) dt ,
  \end{split}
\end{align*}
where \(f^j, f^{\ell}, f^z, f^c\) are functions of time \(t\),
treatment \(A_0\) and baseline covariates \(L_0\) as well as the past
of the processes \(N^{\ell}\) and \(N^{z}\), i.e.,
\(\bar{N}^{\ell}(t-) = (N^{\ell}(u) \, : \, u <t)\) and
\(\bar{N}^{z}(t-) = (N^{z}(u) \, : \, u <t)\), via summary functions
\(\kappa^{\ell}\) and \(\kappa^{z}\). We assume that
\(\kappa^{\ell},\kappa^{z}\) are fixed a priori, while
\(f^j, f^{\ell}, f^z, f^c\) may be data-adaptively learned from the
data, e.g., with highly adaptive lasso estimation
\citep{benkeser2016highly,van2017generally,rytgaard2021estimation},
with estimated versions denoted
\(\hat{f}_n^j, \hat{f}_n^{\ell}, \hat{f}_n^z,
\hat{f}_n^c\). Alternatively, Cox regression estimators can be applied
directly, in which case the precise functional form of
\(f^j, f^{\ell}, f^z, f^c\) is specified a priori and split into an
unrestricted (log) baseline hazard and a linear predictor capturing
dependence on baseline covariates and event history.

\subsection{Estimation of clever covariates} 
\label{sec:clever:covar:estimation}

Our proposal for estimation of clever covariates sketches one
particular algorithm, which can be viewed as a foundation on which
more general and computationally efficient solutions can be
built. This algorithm uses that the sample space of
\(\bar{N}^\ell(t-),\bar{N}^z(t-)\) can be partitioned into a finite
collection of disjoint cubes,
\(\cupdot_{s\in\mathscr{S}} \mathcal{N}_s\), indexed by a finite index
set \(\mathscr{S}\), such that
\begin{align}
  {\hat{f}^{\cdot}_n}(t, \kappa^{\ell}(\bar{N}_1^{\ell}(t-)), \kappa^{z}(\bar{N}_1^{z}(t-)), a_0,\ell_0)
  =  {\hat{f}^{\cdot}_n}(t, \kappa^{\ell}(\bar{N}_2^{\ell}(t-)), \kappa^{z}(\bar{N}_2^{z}(t-)), a_0,\ell_0),
  \label{eq:f:n:dependence:grid}
\end{align}
for all
\( (\bar{N}_1^{\ell}(t-), \bar{N}_1^{z}(t-)), (\bar{N}_2^{\ell}(t-),
\bar{N}_2^{z}(t-))\in \mathcal{N}_s\) and fixed \(t, a_0, \ell_0,
j\). This can be implemented via a matrix-based backward recursion,
which produces estimates of all clever covariates, as well as
estimates of the auxiliary and target parameters. We summarize a
practical, reproducible algorithm in Supplementary Appendix G.

\subsection{Estimation of clever weights}
\label{sec:estimation:clever:weights}
  
Estimation of \(w^a_t\) is routine (see also next Section
\ref{sec:estimation:TMLE:step}). We can estimate \(w^{\alpha}_t\) as follows
\begin{align}
  w^{\alpha}_t (\hat{\Lambda}_n^z) (O)
  & = 
    \frac{ \Prodi_{s < t} \big( \alpha
    \hat{\Lambda}_n^z (ds \mid \F_{s-})\big)^{N^z(ds)}
    \big(1- \alpha\hat{\Lambda}_n^z (ds \mid \F_{s-})\big)^{1-N^z(ds)}}{
    \Prodi_{s < t} \big( \hat{\Lambda}_n^z (ds \mid \F_{s-})\big)^{N^z(ds)}
    \big(1- \hat{\Lambda}_n^z (ds \mid \F_{s-})\big)^{1-N^z(ds)}} \notag\\
  & =
  \alpha^{ N^z(t-)}    \exp \bigg( - (\alpha -1) \int_0^t  \hat{\Lambda}_n^z (ds \mid \F_{s-})\bigg). 
    \label{eq:w2:estimation}
\end{align}
{We emphasize the key simplification in the weight
  \( w^{\alpha}_t (\Lambda^z)\) shown above: the ratio of intensities
  reduces to a constant. This is a feature of our intervention
  approach, and contributes to the robustness of the following
  estimation.}

\subsection{Targeting (TMLE) algorithm}
\label{sec:estimation:TMLE:step}

For each intensity \(\Lambda^\cdot\), with corresponding intensity
process \(\lambda^{\cdot}\), to be targeted, we can define the
intercept-only submodel
\begin{align}
  \Lambda_{\eps}^{\cdot} (dt\mid \F_{t-}) =
  \Lambda^{\cdot} (dt\mid \F_{t-}) \exp( \eps), \quad \eps\in\R,
  \label{eq:Lambda:submodel}
\end{align}
and further the log-likelihood loss function
\begin{align*}
&  \mathscr{L}_{\cdot}(\Lambda^{\cdot} ) (O) = \int_0^{\tau} w^a_t (\pi, \Lambda^c) (O) w^{\alpha}_t (\Lambda^z) (O) h^{x,\cdot}_t (\lbrace\Lambda^j\rbrace_{j=1}^J,\Lambda^{\ell}, \Lambda^{z,\alpha})(O) \log \lambda^{\cdot}(t \mid \F_{t-}) N^{\cdot}(dt) \\
& \qquad\qquad\qquad\qquad\qquad\qquad  - \int_0^{\tau}w^a_t (\pi, \Lambda^c) (O) w^{\alpha}_t (\Lambda^z) (O) h^{x,\cdot}_t (\lbrace\Lambda^j\rbrace_{j=1}^J,\Lambda^{\ell}, \Lambda^{z,\alpha})(O) \Lambda^{\cdot} (dt \mid \F_{t-}). 
\end{align*}
It is straightforward that this pair has the desired property that
\begin{align*}
  \frac{d}{d\eps}\bigg\vert_{\eps = 0}  \mathscr{L}_{\cdot}(\Lambda_\eps^{\cdot} ) (O) =
  \int_0^{\tau} w^a_t (\pi, \Lambda^c) (O) w^{\alpha}_t (\Lambda^z) (O) h^{x,\cdot}_t (\lbrace\Lambda^j\rbrace_{j=1}^J,\Lambda^{\ell}, \Lambda^{z,\alpha})(O) \big( N^{\cdot}(dt) -  \Lambda^{\cdot} (dt \mid \F_{t-})\big). 
\end{align*}
Different versions, or combinations, of including
\( w^a_t, w^{\alpha}_t, h^{x,\cdot}_t\) in a weight
or as a covariate are possible as well. \\

Our targeting procedure consists of the following:
\begin{enumerate}
\item Estimators \(\hat{\Lambda}_n^z\), \(\hat{\Lambda}_n^c\) and
  \(\hat{\pi}_n\) for \(\Lambda^z\), \(\Lambda^c\) and \(\pi\), based on which we get
  an estimator for \(w^a_t(\pi,\Lambda^c)\) and  \(w^{\alpha}_t(\Lambda^z)\) as:
  \begin{align*}
    w^a_t (\hat{\pi}_n,\hat{\Lambda}_n^c)(O) &= 
                                               \frac{\delta_a (A) }{\hat{\pi}_n (A\mid L) \prodi_{s < t} (1- \hat{\Lambda}_n^c(ds \mid  \F_{t-}))}   , \\
     w^{\alpha}_t (\hat{\Lambda}_n^z)(O) &= \alpha^{N^z(t-)}    \exp \bigg( - (\alpha -1)\int_0^t  \hat{\Lambda}_n^z (ds \mid \F_{s-})\bigg).
  \end{align*}
\item Initial estimators \(\hat{\Lambda}_{n}^j\) and
  \(\hat{\Lambda}_{n}^{\ell}\) for \(\Lambda^{j}\) and
  \(\Lambda^{\ell}\), \(j=1,\ldots, J\), based on which we also
  estimate clever covariates \( h^{x,j}_t (\Lambda^d,\Lambda^y, \Lambda^{z,\alpha})\),
  \( h^{x,\ell}_t (\Lambda^d,\Lambda^y, \Lambda^{z,\alpha}) \) and
  \( h^{x,z,\alpha}_t (\Lambda^d,\Lambda^y, \Lambda^{z,\alpha}) \).
\item A targeting procedure to update
  \(\hat{\Lambda}_n^1, \ldots, \hat{\Lambda}_n^J\),
  \(\hat{\Lambda}_n^{\ell}\) and \(\hat{\Lambda}_n^{z}\). We propose
  to execute this in an iterative manner starting with
  \(\hat{\Lambda}_{n,0}^1 := \hat{\Lambda}_{n}^1, \ldots,
  \hat{\Lambda}_{n,0}^J := \hat{\Lambda}_{n}^J \),
  \(\hat{\Lambda}_{n,0}^{\ell} := \hat{\Lambda}_{n}^{\ell}\) and
  \(\hat{\Lambda}_{n,0}^{z} := \hat{\Lambda}_{n}^{z}\), and the
  \(m\)th step updating, for each \(x'=1,\ldots,J,\ell,z\),
  \(\hat{\Lambda}_{n,m}^{x'} \mapsto \hat{\Lambda}_{n,m+1}^{x'}\)
  along the parametric submodel \eqref{eq:Lambda:submodel} to solve
  the relevant term corresponding to one of
  \eqref{eq:effi:if:1}--\eqref{eq:effi:if:z} equal to zero for given
  \(w^a_t(\hat{\Lambda}_n^c, \hat{\pi}_n)\),
  \( w^{\alpha}_t (\hat{\Lambda}_{n,m}^z)\), and
  \( h^{x',1}_t (\hat{\Lambda}_{n,m}^{1},\ldots,
  \hat{\Lambda}_{n,m}^{J}, \hat{\Lambda}_{n,m}^{\ell},
  \hat{\Lambda}_{n,m}^{z,\alpha})\) (\( h^{x,z,\alpha}_t \) for
  \(x'=z\)).  This procedure is repeated until
  \(\vert \, \mathbb{P}_n \phi(\hat{P}_n^{*}) \, \vert \le s_n\),
  where
  \(\hat{P}_n^{*}= (\hat{\Lambda}^{1}_{n,m^*}, \ldots,
  \hat{\Lambda}^{J}_{n,m^*},
  \hat{\Lambda}^{\ell}_{n,m^*},\hat{\Lambda}^z_{n,m^*},
  \hat{\Lambda}^c_n, \hat{\pi}_n)\) and
  \(s_n = \sqrt{\mathbb{P}_n (\phi ( \hat{P}_n))^2} / (n^{1/2} \log
  n)\) with \(\mathbb{P}_n ( \phi ( \hat{P}_n))^2\) estimating the
  variance of the efficient influence curve based on the collection of
  initial estimators for the nuisance parameters
  \(\hat{P}_n=(\hat{\Lambda}^{1}_{n}, \ldots, \hat{\Lambda}^{J}_{n},
  \hat{\Lambda}^{\ell}_{n},\hat{\Lambda}^z_n, \hat{\Lambda}^c_n,
  \hat{\pi}_n)\).
\end{enumerate}

\section{Simulation study}
\label{sec:sim:study}

For illustration and proof of concept, this section evaluates
finite-sample estimation behavior and inferential properties in the
setting that mirrors the operation/surgery example considered in
Section~\ref{sec:illustration:application}. We consider estimation and
inference for both $\alpha$-specific parameters, examples of
calibrated/composite estimands obtained by solving for specific
equations regarding the auxiliary risk, and various contrasts of
practical interest. The experiments demonstrate the practical use of
the proposed estimands and TMLE-based estimators, while also
illustrating implications of practical positivity violations,
manifested as extreme clever weights, that arise when interventions
are pushed far beyond the observed data support.  Source code to
reproduce the results of this section is available from GitHub
(\url{https://github.com/helenecharlotte/Web-appendix-calibrated-intensity-interventions}).

Data are simulated as follows.  A baseline covariate is generated as
$L_0 \sim \mathrm{Unif}(0,1)$. There is no baseline
treatment. Conditional on the observed history, each counting process
is generated from a multiplicative intensity model of the form
\begin{align*}
  \lambda^x(t \mid \mathcal{F}_{t-}) 
  =
  \lambda_{\mathrm{baseline}}^x(t)
  \exp\!\left(
    \beta^x_{L_0} L_0
    + \beta^x_{z} N^z(t-)
    + \beta^x_{\ell} N^{\ell}(t-)
  \right),
\end{align*}
where $\mathcal{F}_{t-}$ denotes the observed filtration just prior to
time $t$. The baseline hazards are Weibull,
\( \lambda_{\mathrm{baseline}}^x(t) = \eta^x \nu^x t^{\nu^x - 1}, \)
with parameters $\eta^x>0$ and $\nu^x>0$.  The specific parameter
values used in the simulations are: \(\beta_{L_0}^a = 1\),
\(\beta_{z}^1 = -0.5\), \(\beta_{z}^l = -2.5\),
\(\beta_{\ell}^z = 3\), \(\beta_{\ell}^1 = 2.5\), \(\eta^z = 0.085\),
\(\eta^{\ell} = 0.1\), \(\eta^1 = 0.025\), \(\eta^0 = 0.1\), and
\(\nu^x=1.1\).  These values coincide with those used in the
illustrative example of Section~\ref{sec:illustration:application}. In
this setting, $N^{\ell}$ governs disease progression, $N^{z}$ governs
the operation, $N^{1}$ governs all-cause mortality, and $N^{c}$
governs administrative censoring.

We consider the following classes of target parameters, for which we
assess point estimation, uncertainty quantification, and hypothesis
testing:
\begin{enumerate}
\item $\alpha$-specific auxiliary and target parameters
  $\Psi_z^{\alpha}(P)$ and $\Psi_1^{\alpha}(P)$ over the grid of
  values \(\alpha \in \{0,\,0.1,\,0.25,\,0.5,\,1,\,1.5,\,2.25,\,3\}\).
  For each \(\alpha\), we compute targeted and non-targeted estimators
  based on both correctly specified models and misspecified models for
  \(\lambda^1\) (leaving out the dependence on \(N^{\ell}(t-)\)) and
  \(\lambda^{\ell}\) (leaving out the dependence on \(N^z(t-)\)).  The
  non-targeted estimators are obtained by fitting parametric
  multiplicative intensity models and evaluating the target functional
  at the fitted distribution, and targeted estimators are obtained by
  updating the initial parametric fits through the targeting step.
\item Calibrated parameters \(\alpha^{\theta}(P)\) and corresponding
  composite parameters
  \(\Psi_1^{\theta}(P) = \Psi_1^{\alpha^{\theta}(P)}(P)\) targeting
  specific levels \(\Psi^{\alpha}_z (P) = \theta\) of auxiliary risk,
  for levels
  \(\theta \in \lbrace 0.25, 0.333, 0.5, 0.667,
  0.75\rbrace\). Estimation uses the same TMLE construction described
  above with an additional root-finding step to estimate
  $\alpha^{\theta}(P)$, and an additional TMLE step to target the
  estimator for the composite parameter.
\item Four different contrasts (differences in outcome risks), and
  corresponding power of tests of no effect:
  \begin{enumerate}
  \item \(\Psi^{\alpha=0}_1 (P) \) to \(\Psi^{\alpha=1}_1 (P) \)
    (comparing `censoring for \(z\) events' to observed occurrence of
    \(z\) events).
  \item \(\Psi^{\theta = 0.25}_1 (P) \) to 
    \(\Psi^{\theta =0.75}_1 (P) \) (comparing the outcome risk if 25\%
    had been operated to the outcome risk 75\% been operated).
  \item \(\Psi^{\rho = 0.6}_1 (P) \) to \(\Psi^{\alpha=1}_1 (P) \)
    (comparing the outcome risk if 40\% less had been operated to
    observed occurrence of type \(z\) events).
  \item \(\Psi^{\rho = 1.5}_1 (P) \) to \(\Psi^{\alpha=1}_1 (P) \)
    (comparing the outcome risk if 50\% more had been operated to
    observed occurrence of type \(z\) events).
\end{enumerate}
\end{enumerate}

Inference for all parameters is based on Wald confidence intervals and
Wald tests using the standard error estimated using the empirical
variance of the estimated influence curve. For reference, we also
report `oracle' coverage probabilities obtained by using the Monte
Carlo standard deviation (i.e., the empirical standard deviations
across simulation repetitions).
Each scenario is replicated $M=500$ times with sample size $n=1000$.
No truncation of clever weights is applied; instead we report weight
quantiles to diagnose numerical instability and identify regimes
(values of \(\alpha\)) where some sort of truncation would be
warranted.

Results are summarized in Figures
\ref{fig:operation:simulation:results:alpha}--\ref{fig:operation:simulation:results:contrasts}. Detailed
summaries, including bias, mean squared error, standard errors, and
coverage probabilities for all estimators, are reported in tables in
Supplementary Appendix H.
Figure~\ref{fig:operation:simulation:results:alpha} displays the
estimated $\alpha$-specific outcome and auxiliary target parameters
across the considered grid of \(\alpha\)-values.  The left column
shows results under correct specification and the right column the
misspecified case.  The targeting step is seen to remove the heavy
bias visible in the initial plug-in estimates in the misspecified
scenario as expected from the double robustness properties. The
targeted estimator exhibits negligible bias and near-nominal
confidence interval coverage across the considered values of
\(\alpha\).  For low $\alpha$ (strongly down-weighting operation
occurrence), however, confidence intervals are somewhat
undercovering. This is driven by severe practical positivity issues
resulting in extreme inverse probability weights, as shown in
Figure~\ref{fig:operation:simulation:results:weights} which plots the
empirical distribution of extreme inverse probability weight quantiles
across simulation repetitions.  As expected, weights become more
variable and extreme for values of $\alpha$ far below one, reflecting
intervention regimes where estimation becomes numerically more
challenging.

\begin{figure}[!ht]
\centering
 \begin{center} 
    \includegraphics[width=1\textwidth,angle=0]{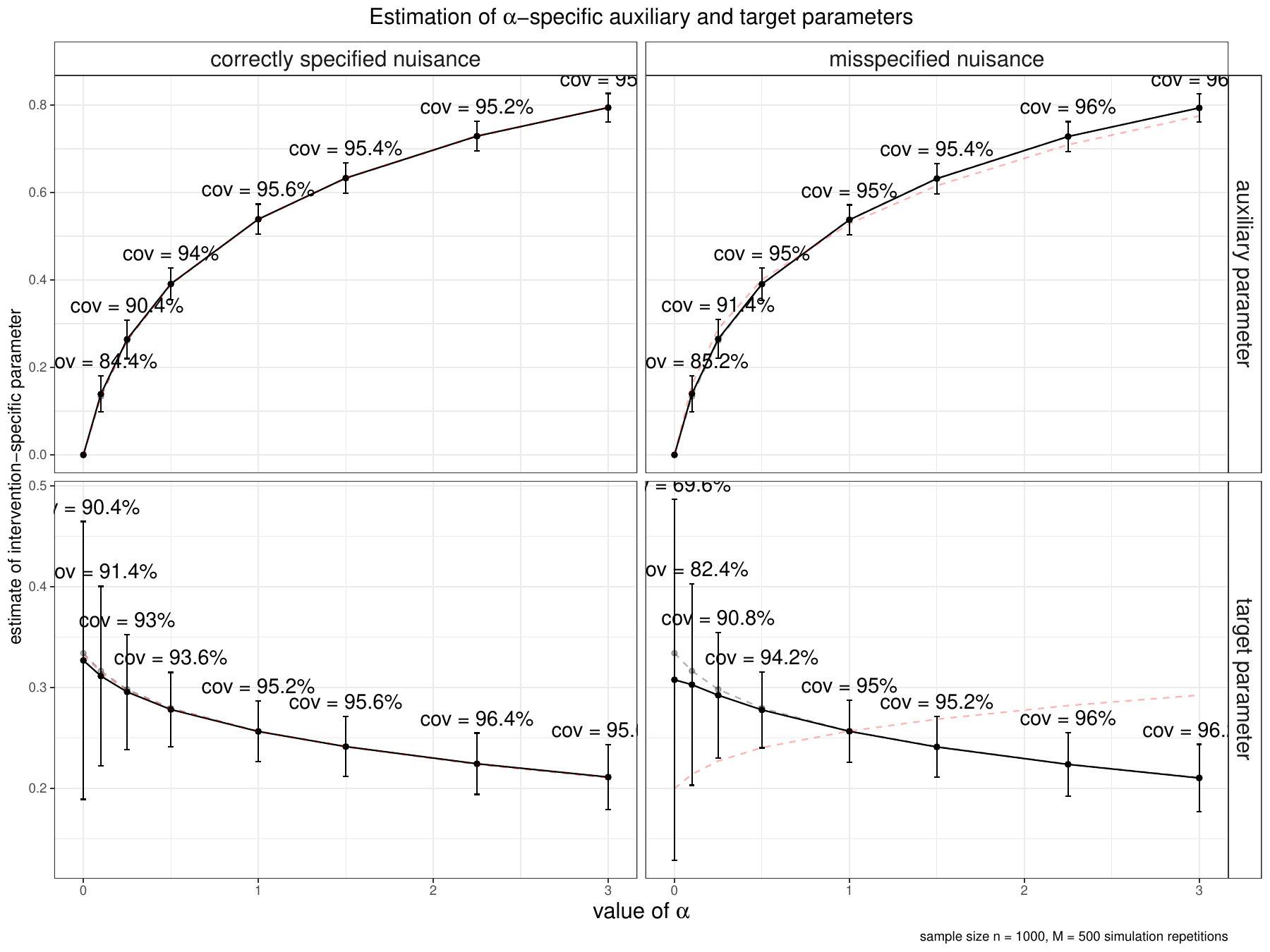}
  \end{center}  
  \vspace{-0.5cm}    
  \caption{Estimated \(\alpha\)-specific auxiliary ($\Psi_z^{\alpha}$,
    upper panels) and target ($\Psi_1^{\alpha}$, lower panels)
    parameters across a grid of values of $\alpha$.  Left column:
    correctly specified nuisance models. Right column: misspecified
    outcome and covariate intensities.  Red dashed lines indicate
    initial plug-in estimates; solid lines are TMLE updates.  For
    small $\alpha$ (strongly down-weighting operation occurrence)
    estimation quality deteriorates because a few individuals receive
    very large inverse probability weights when we attempt to enforce
    an intervention that is unrealistic for them; this produces
    inflated variance and undercoverage of confidence intervals.
  }
   \label{fig:operation:simulation:results:alpha}  
 \end{figure}

 \begin{figure}[!ht]
\centering
 \begin{center}
    \includegraphics[width=1\textwidth,angle=0]{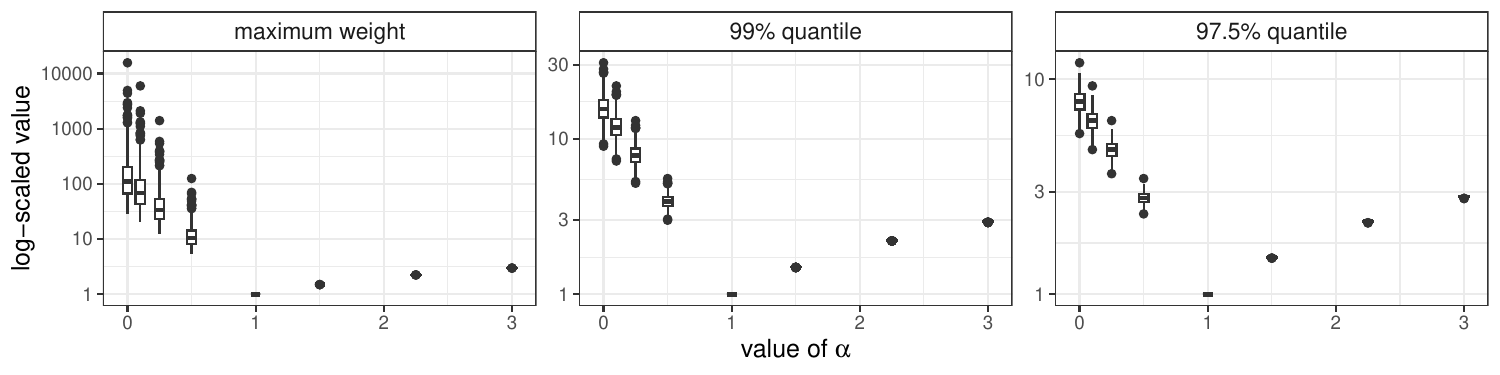}
  \end{center}  
  \vspace{-0.5cm}    
  \caption{Quantiles (97.5\%, 99\%, 100\%) of inverse probability weights across
Monte Carlo replicates for each $\alpha$ (vertical axis on a log
scale).  Extremely large weights appear for small $\alpha$, signaling
severe practical positivity violations in those intervention regimes.  Investigation
shows these extreme weights arise from a small subset of individuals
who, according to their covariate/history profile, are very unlikely
never to undergo the operation; forcing $\alpha$ near zero effectively
requires unlikely counterfactual behaviour for these subjects and produces
unstable estimates.   
These diagnostics could be consulted before interpreting contrasts:
regimes with extreme quantiles indicate estimands that rely on
extrapolation which should be avoided.  }
   \label{fig:operation:simulation:results:weights}  
 \end{figure}

 Figure~\ref{fig:operation:simulation:results:theta} shows the
 estimated calibration parameters $\alpha^{\theta}(P)$ and composite
 parameters $\Psi_1^{\theta}(P)$ for levels
 $\theta\in\{0.25,0.333,0.5,0.667,0.75\}$.  Lower values of $\theta$
 correspond to smaller calibrated $\alpha$ and are harder to estimate
 due to extreme weights. This appears most harmful for the calibration
 parameter itself, which we emphasize is typically not of primary
 inferential interest, while the resulting composite parameter
 estimates remain relatively well-behaved.

\begin{figure}[!ht] 
\centering 
 \begin{center} 
    \includegraphics[width=1\textwidth,angle=0]{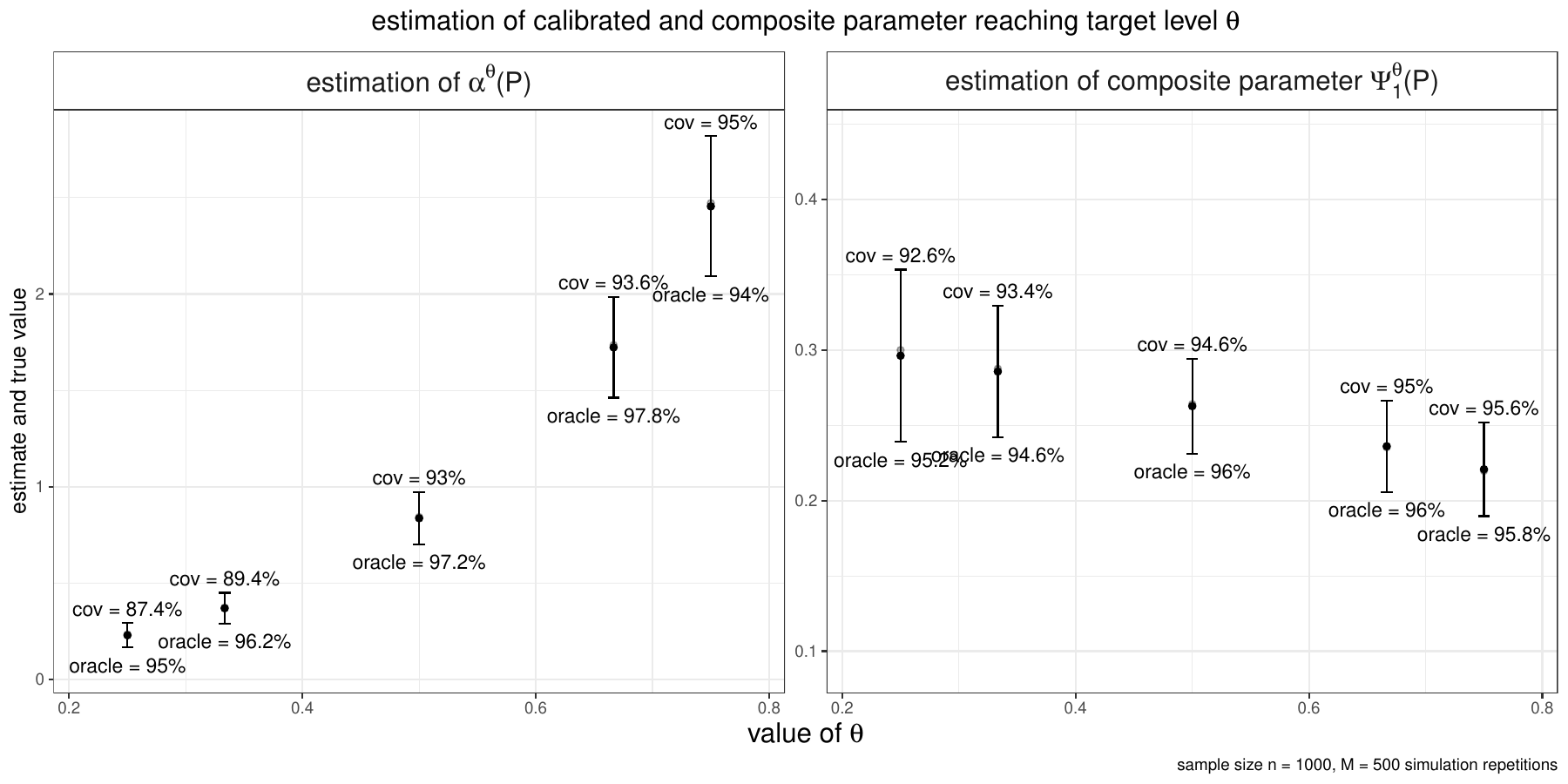}
  \end{center}  
  \vspace{-0.5cm}    
  \caption{Estimates of calibrated (left) and composite (right)
    parameters \(\alpha^{\theta}(P)\) and \(\Psi^{\theta}_1(P)\).
    Lower $\theta$ produce smaller calibrated $\alpha$ and are
    associated with greater numerical instability.}
   \label{fig:operation:simulation:results:theta}  
 \end{figure}

 Figure~\ref{fig:operation:simulation:results:contrasts} displays
 Monte Carlo distributions and empirical rejection rates for the four
 contrasts involving different combinations of \(\alpha\)-specific and
 calibrated parameters.  An important takeaway is that power is not
 determined by effect size alone.  Contrasts involving extreme
 intervention regimes (corresponding to very small values of $\alpha$)
 suffer from severe practical positivity violations and inflated
 variance, which substantially reduces power despite large nominal
 shifts in the intervention.  By contrast, contrasts comparing
 relevant and realistic intervention regimes, those that avoid extreme
 weights and remain close to the observed data support, achieve high
 power at $n=1000$ even when the corresponding effect sizes are more
 moderate.  Overall, these results highlight that well-powered causal
 comparisons are obtained by balancing substantive relevance of the
 intervention with statistical feasibility, rather than by targeting
 extreme intervention regimes.

 \begin{figure}[!ht]
\centering 
 \begin{center}
    \includegraphics[width=0.7\textwidth,angle=0]{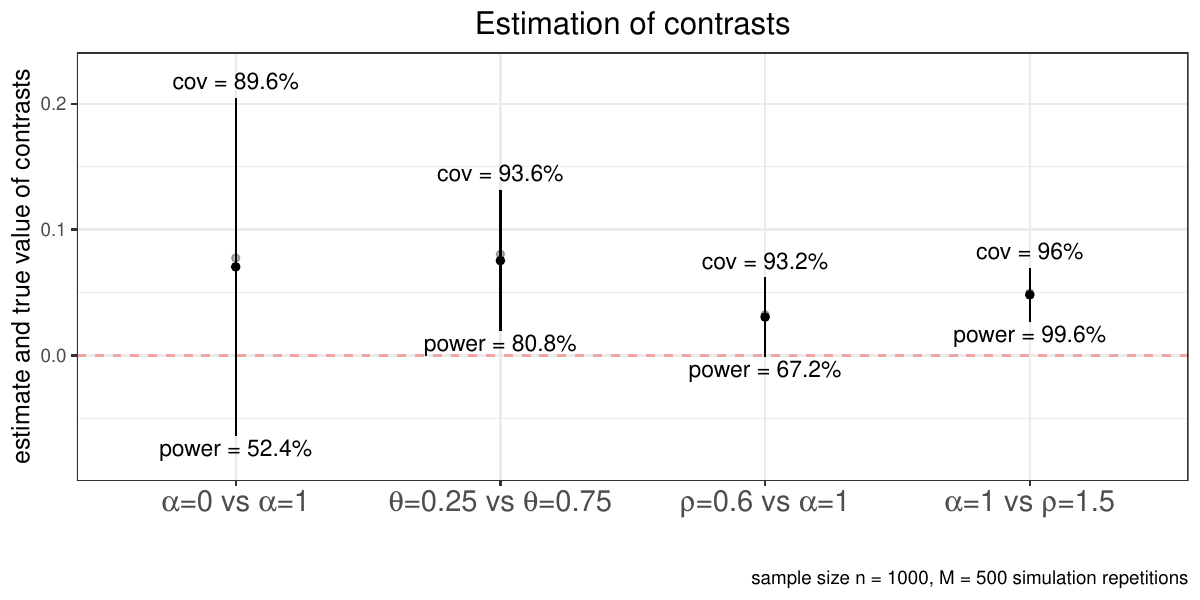}
  \end{center}   
  \vspace{-0.5cm}      
  \caption{Estimates and empirical power for four contrasts: (a) 
    $\Psi^{\alpha=0}_1$ vs.\ $\Psi^{\alpha=1}_1$; (b)
    $\Psi^{\theta=0.25}_1$ vs.\ $\Psi^{\theta=0.75}_1$; (c)
    $\Psi^{\rho=0.6}_1$ vs.\ $\Psi^{\alpha=1}_1$; (d)
    $\Psi^{\rho=1.5}_1$ vs.\ $\Psi^{\alpha=1}_1$.  Distributions
    across Monte Carlo replicates are shown together with empirical
    rejection rates (power) from Wald tests.}
   \label{fig:operation:simulation:results:contrasts}  
 \end{figure}

\section{Discussion}
\label{sec:discussion}

This work presents a general inferential framework for stochastic
\(\alpha\)-scaled intensity interventions and introduces calibrated
interventions that connect abstract intervention parameters to
clinically interpretable quantities.  By scaling intensities rather
than imposing static or deterministic rules, the considered approach
respects natural heterogeneity, avoids unrealistic interventions, and
provides a principled way to evaluate how shifts in intermediate
processes propagate to final outcomes. Calibrating these interventions
to benchmarks such as risk levels or subgroup differences further
strengthens interpretability, defining statistical parameters that are
directly relevant in applied contexts.
Overall, the intervention definition simplifies both positivity
assumptions and the estimation of clever (inverse probability)
weights, mitigating a common source of
instability. 

Our formulation connects closely to incremental propensity score
interventions in discrete time, but differs by operating directly on
the intensity scale and by accommodating general event-history data
with competing risks and right censoring. This makes it more broadly
applicable, but also highlights that it is not simply an extension of
existing frameworks.
Calibrated interventions, which we propose as a novel additional
layer, ensures that the resulting parameters remain closely tied to
subject-matter questions, and allows investigators to tailor analyses
to clinically meaningful benchmarks.

Beyond absolute risk as a calibration benchmark, other clinically
relevant summaries, such as years of healthy life lost due to specific
events \citep{andersen2013decomposition}, could equally serve as
targets. The calibration perspective might also be inverted: instead
of asking what level of mediator reduction corresponds to a given risk
decrease, one may ask how much treatment uptake would need to increase
to achieve a specified reduction in mortality. Additional extensions
include intervening not only on the overall intensity of mediator
events but also on their dependence on specific biomarkers. For
example, if treatment initiation depends strongly on whether a
biomarker crosses a threshold, one could study interventions that
scale this dependence up or down. 

Finally, while we outline a principled TMLE based estimation strategy
that supports flexible nuisance estimation, the generality of the
proposed intervention framework naturally admits a range of
alternative and more specialized implementation choices. Future work
may therefore explore more general or computationally efficient
estimation procedures, particularly in settings with high-dimensional
histories or complex event structures, without altering the underlying
target parameters or inferential framework developed here.

\newpage

\bibliographystyle{abbrvnat}

\section{Competing interests}
No competing interest is declared.



\section{Acknowledgements}

This research is partially funded by the Novo Nordisk Foundation grant
NNF23OC0084961. Additional funding was provided by a philanthropic
gift from Novo Nordisk.

\endgroup 

\end{document}